\documentclass[final,twocolumn,twoside]{IEEEtran}
\usepackage{amsmath,amsfonts,amssymb,amsbsy,bm,paralist,theorem,color}
\usepackage{graphicx,algorithmic,algorithm}
\usepackage{cite}
\usepackage{url}
\newtheorem{lem}{Lemma}
\newtheorem{rem}{Remark}
\newtheorem{Proposition}{Proposition}

\newcommand{\Rset}{\mathbb{R}}

\newcommand{\Cset}{\ensuremath{\mathbb{C}}}
\newcommand{\rank}{\rm{rank}}
\newcommand{\SNR}{{\rm{SNR}}}
\newcommand{\SINR}{{\rm{SINR}}}
\DeclareMathOperator*{\st}{s.t.}
\DeclareMathOperator*{\Tr}{Tr}

\begin{document}
\title{Joint Beamforming and Power Splitting Control in Downlink Cooperative SWIPT NOMA Systems}
\author{
Yanqing Xu, Chao Shen, Zhiguo Ding, Xiaofang Sun, Shi Yan, Gang Zhu, and Zhangdui Zhong\\

\thanks{Manuscript received November 26, 2016; revised March 22, 2017; accepted May 24, 2017. Date of publication; date of current version.
The associate editor coordinating the review of this manuscript and approving it for publication was Prof. Xin Wang.
The work of Y. Xu and C. Shen are supported by NSFC 61501024, RCS2016ZZ004, ISN14-09,
the Fundamental Research Funds for the Central Universities (2017JBM315), and
National S\&T Major Project 2016ZX03001021-003.
The work of Z. Ding is supported by the UK EPSRC under grant number EP/L025272/1 and by H2020-MSCA-RISE-2015 under grant number 690750.
The work of Z. Zhong is supported by NSFC U1334202.
Part of this work was presented at the IEEE International Conference on Communications (ICC), Paris, May. 2017 \cite{Conf_version}. (\textit{Corresponding Author: Chao Shen.})

Y. Xu and C. Shen are with the State Key Laboratory of Rail Traffic Control and Safety, Beijing Jiaotong University, Beijing 100044, China, and also with the State Key Laboratory of Integrated Services Networks, Xidian University, Xi'an  710071, China (e-mail: \{xuyanqing,chaoshen\}@bjtu.edu.cn).

Z. Ding and S. Yan are with the School of Computing and Communications, Lancaster University, Lancaster LA1 4WA, UK (e-mail: \{z.ding, s.yan3\}@lancaster.ac.uk).

X. Sun, G. Zhu and Z. Zhong are with the State Key Laboratory of Rail Traffic Control and Safety, Beijing Jiaotong University, Beijing 100044, China, and also with the Beijing Engineering Research Center of High-speed Railway Broadband Mobile Communications (e-mail:\{sunxiaofang,gzhu,zhdzhong\}@bjtu.edu.cn).

{\bf Bib information}:
Y. Xu, C. Shen, Z. Ding, X. Sun, S. Yan, G. Zhu, and Z. Zhong, ``Joint Beamforming
Design and Power Splitting Control in Cooperative SWIPT NOMA
Systems,'' IEEE Trans. Signal Process., June. 2017.
}}
\maketitle


\begin{abstract}
This paper investigates the application of simultaneous wireless information and power transfer (SWIPT) to cooperative non-orthogonal multiple access (NOMA).
A new cooperative multiple-input single-output (MISO) SWIPT NOMA protocol is proposed, where a user with a strong channel condition acts as an energy-harvesting (EH) relay to help a user with a poor channel condition. The power splitting (PS) scheme is adopted at the EH relay. By jointly optimizing the PS ratio and the beamforming vectors, the design objective is to maximize the data rate of the ``strong user'' while satisfying the QoS requirement of the ``weak user''. It boils down to a challenging nonconvex problem. To resolve this issue, the semidefinite relaxation (SDR) technique is applied to relax the quadratic terms related with the beamformers, and then it is solved to its global optimality by two-dimensional exhaustive search. We prove the rank-one optimality, i.e., the SDR tightness, which establishes the equivalence between the relaxed problem and the original one. To further reduce the high complexity due to the exhaustive search, an iterative algorithm based on successive convex approximation (SCA) is proposed, which can at least attain its stationary point efficiently. In view of the potential application scenarios, e.g., Internet of Things (IoT), the single-input single-output (SISO) case of the cooperative SWIPT NOMA system is also studied. The formulated problem is proved to be strictly unimodal with respect to the PS ratio. Hence, a golden section search (GSS) based algorithm with closed-form solution at each step is proposed to find the unique global optimal solution. It is worth pointing out that the SCA method can also converge to the optimal solution in SISO cases. In the numerical simulation, the proposed algorithm is numerically shown to converge within a few iterations, and the SWIPT-aided NOMA protocol outperforms the existing transmission protocols.
\end{abstract}


\section{Introduction}
Recently, great interests have been drawn in non-orthogonal multiple access (NOMA), which is envisioned as an enabling technique to improve the spectral efficiency (SE) of the upcoming fifth generation (5G) network \cite{NOMA-2013}, especially in order to satisfy the requirement of the Internet of Things (IoT) to support massive connectivity
\cite{Ding-Magazine}. The key idea of NOMA is to realize multiple access in the power domain, which allows multiple users to be concurrently severed at the same resource elements, e.g., time slots, frequency bands and spreading codes.
By applying power-domain multiplexing at the transmitter and successive interference cancellation (SIC) at the receivers, the NOMA scheme can achieve a better system performance compared with the conventional orthogonal multiple access (OMA) scheme 
\cite{tse-2005,Ding-2014-SPL,Ding-TVT}. %
Due to the advances in multiple-input multiple-output (MIMO) techniques, the combination of NOMA and MIMO has been investigated to improve the system SE further \cite{Ding-MIMO-TWC2016, Sun-2015-WCL, Lan-2014}. A multiple-input single-output (MISO) NOMA scheme was considered in \cite{Ding-2016-TSP}, in which the downlink sum rate is maximized by applying the minorization-maximization (MM) method. A relay-aided cooperative NOMA scheme \cite{Ding-2015-CL} was proposed to improve the communication reliability of the cell-edge users, wherein, the cell-center users act as relays to help the cell-edge users.

In addition to the requirement of high SE, the energy efficiency (EE) is also a key performance indicator in 5G. As a promising solution to improve EE, scavenging energy from the environment to power the communication system has been studied intensively in the past several years \cite{Sennur-2011-JSAC, Sun-2014}. Since radio-frequency (RF) signals can carry both information and energy, simultaneous wireless information and power transfer (SWIPT) was introduced initially in \cite{Varshney-2008-ISIT} and attracted increasing attention in the communication research community \cite{Zhang-2013-TWC, Shen-2014TSP}.
Specifically, to handle the limitation of the state-of-the-art circuit design, two practical receiver architectures, namely, the time switching (TS) receiver and power splitting (PS) receiver, were investigated in \cite{Zhang-2013-TWC}.
The transmitter design for SWIPT in a MISO interference channel is considered in \cite{Shen-2014TSP},  and  therein the sum rate maximization problem was studied under four practical receiver architectures. In order to exploit the diversity gain of the cooperative transmission, many efforts have been directed towards the relay-aided SWIPT systems using amplify-and-forward (AF) \cite{Nasir-2013-TWC}, or decode-and-forward (DF) \cite{Ding-2014-TWC} protocols.

Motivated by the requirements of 5G and the advantages of NOMA and SWIPT, \cite{Ding-2015-CL} proposed a cooperative NOMA strategy to improve the communication reliability for the cell-edge users. However, due to the limited energy storage at the relay nodes in, e.g., the IoT scenarios, there exists a tradeoff between the information receiving for themselves and information forwarding for others. Thus, the idea of SWIPT is introduced to the NOMA system to alleviate the energy constraint in \cite{Liu-2016}, in which the cell-center users act as energy harvesting (EH) relays to help the cell-edge users and the system outage performance is analyzed with three user-paring schemes in a single-input single-output (SISO) scenario.

In this paper, we propose a novel cooperative SWIPT NOMA protocol and focus on the optimal transmission design.
The considered system consists of a base station (BS) and two users, in which the ``strong user" harvests the energy by using power splitting scheme, and then acts as a relay to help the ``weak user".
Two cases, namely the MISO case and SISO case of the cooperative SWIPT NOMA system, are considered.
In the MISO case, the joint design of beamforming and power splitting is considered.
Note that, in some potential application scenarios, e.g., IoT or massive machine type communications, the SISO antenna configuration still has its advantages, e.g., low power and low cost. Thus it motivates us to investigate the SISO case of the cooperative SWIPT NOMA protocol design.

The main contributions of this paper are summarized as follows.
\begin{itemize}
  \item We propose a cooperative SWIPT-aided NOMA transmission strategy, where the ``strong user'' acts as an EH relay to help the ``weak user'' to improve the communication reliability. By utilizing the power splitting scheme, the ``strong user'' forwards the information
      by using the harvested energy only. As a result, relay transmission is powered by the harvested energy, and doesn't consume energy from its battery in the information forwarding stage.

  \item The data rate maximization of the ``strong user'' is considered in MISO cases. The formulated problem is nonconvex and challenging to solve. We then apply the SDR technique to relax the quadratic terms related with the beamformers. The SDR tightness is proved, which establishes the equivalence between the two problems. However, due to the high complexity to get the global optimal solution, an SCA-based iterative algorithm is developed to efficiently obtain at least a stationary point of the problem.
  \item Motivated by the practical applications, we consider the cooperative SWIPT NOMA transmission strategy in SISO cases.
        We prove that the optimal value of the formulated problem is strictly unimodal with respect to the PS ratio and thus the optimal PS ratio can be efficiently obtained through the GSS method. Thus an iterative algorithm is presented based on the GSS method, and at each iteration, the optimal power allocation coefficient has a closed-form expression.
        In addition, we found that both the SCA- and the GSS-based algorithms can converge to the unique global optimum in SISO cases.
        The performance advantages of the proposed transmission strategy are shown in the simulation section.
\end{itemize}

The rest of the paper is organized as follows. In Section II, the system model and the problem formulation of the cooperative SWIPT NOMA in MISO cases are presented. In Section III,
an SCA-based iterative algorithm is proposed to solve the joint beamforming design and power splitting control problem. The cooperative SWIPT NOMA in SISO cases are analyzed in Section IV and a semiclosed-form optimal solution is derived. Simulation results and conclusions are given in Section V and Section VI, respectively.

\textit{Notations:} Column vectors and matrices are denoted by boldfaced lowercase and uppercase letters, e.g., ${\bf a}$ and ${\bf A}$. The superscript $(\cdot)^T$ and $(\cdot)^H$ denote the transpose and (Hermitian) conjugate transpose, respectively. $\Tr({\bf A})$ and rank$({\bf A})$ represent the trace and the rank of matrix ${\bf A}$, respectively. $\mathbb{E}(\cdot)$ denotes the statistical expectation. $\|{\bf a}\|$ denotes the Euclidean norm of vector ${\bf a}$ and $|b|$ denotes the magnitude of a complex
number $b$.


\section{System Model and Problem Formulation}

\begin{figure}[!t]
  \centering
  \includegraphics[width=0.94\linewidth]{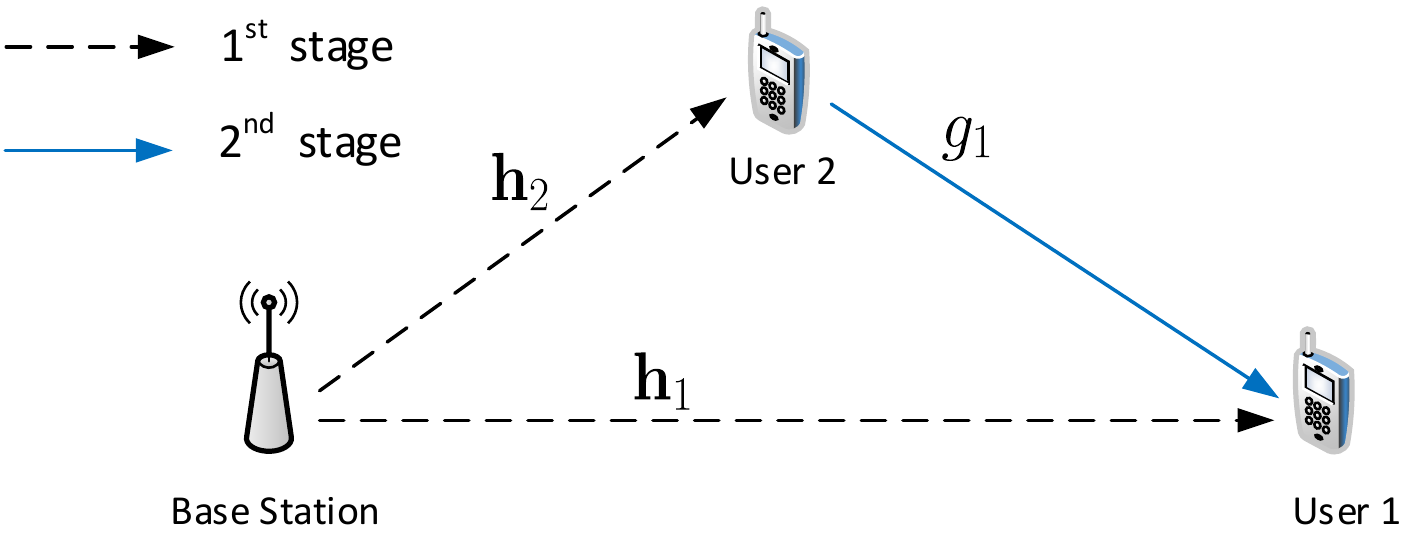}\\
  \caption{A two-stages downlink MISO transmission system.}
\end{figure}

A downlink MISO transmission system is considered, as shown in Fig. 1, where the BS is equipped with $N_t$ antennas and there are two single-antenna users, i.e., user $1$ and user $2$.
Without loss of generality, we assume that these two users have different channel conditions.
It is assumed that user $1$ has a poorer channel condition than user $2$, e.g., user $1$ is a cell-edge user and user $2$ is a cell-center user.
Hence, to guarantee the QoS requirement of user $1$, user $2$ can act as an EH relay to help user $1$.

Two stages are involved in the cooperative SWIPT NOMA transmission.
In the first stage, user 1 receives the signal while user 2 performs SWIPT, i.e., the signal received at user $2$ is split into two parts -- one for information decoding and the other for energy harvesting.
Then in the second stage, user 2 forwards the message to user 1 with the harvested energy, while user 1 combines the message received in the two stages by maximal-ratio combination (MRC) and then decodes it.
The process is detailed as below.

\subsection{Direct Transmission Stage}
In this stage, the transmitted signal at the BS is ${\bf x}= {\bf w}_1 x_1 + {\bf w}_2 x_2$ \cite{Sun-2015-WCL}, where $x_1,x_2\in\Cset$ are the i.i.d. information bearing messages for user $1$ and user $2$, respectively. The power of the transmitted symbol is normalized, i.e.,  $\mathbb{E}|x_1|^2 = \mathbb{E}|x_2|^2 = 1$, and ${\bf w}_1, {\bf w}_2  \in \mathbb{C}^{N_t}$ are the corresponding transmit beamformers, satisfying the power constraint, i.e., $\|{\bf w}_1\|^2+\|{\bf w}_2\|^2 \le 1$.  Then, the observation at user $1$ is given by
\begin{align}
y_1^{(1)} = \sqrt{P_s} {\tilde{\bf h}}_1^{H} ({\bf w}_1 x_1 + {\bf  w}_2 x_2) + z_1^{(1)},
\end{align}
where
$P_s$ is the transmit power at the BS,
$\tilde{{\bf h}}_1\in\mathbb{C}^{N_t}$ is the channel coefficient between the BS and user $1$, and
$z_1^{(1)} \sim \mathcal{CN}(0,\sigma_1^2)$ is the additive Gaussian white noise (AWGN) at user $1$.
Then, the received signal-to-interference-plus-noise-ratio (\SINR) at user $1$ to detect $x_1$ can be expressed as
\begin{align}\label{sinr1a}
\SINR_{1}^{(1)} = \frac{P_s | \tilde{{\bf h}}_1^{H} {\bf  w}_1|^2}{P_s |\tilde{{\bf h}}_1^{H} {\bf  w}_2|^2 +\sigma_1^2} = \frac{ | {\bf h}_1^{H} {\bf  w}_1|^2}{ |{\bf h}_1^{H} {\bf  w}_2|^2 + 1},
\end{align}
where ${\bf h}_1 \triangleq \frac{ \sqrt{P_s}}{\sigma_1} \tilde{{\bf h}}_1$.

For user $2$, the power splitting architecture, as depicted in Fig. \ref{swipt}, is employed to perform SWIPT. Then, the received signal for information decoding at user $2$ can be described as
\begin{align}
y_2^{(1)} = \sqrt{P_s}\sqrt{1 - \beta}\tilde{{\bf h}}_2^H({\bf w}_1 x_1 + {\bf w}_2x_2) + z_2^{(1)},
\end{align}
where $\beta\in[0,1]$ is the PS ratio to be optimized later, $\tilde{{\bf h}}_2 \in \mathbb{C}^{N_t}$ is the channel coefficient between the BS and user $2$,
and $z_2^{(1)} \sim \mathcal{CN}(0,\sigma_2^2)$ is the AWGN.

\begin{figure}[!tp]
  \centering
  \includegraphics[width=0.95\linewidth]{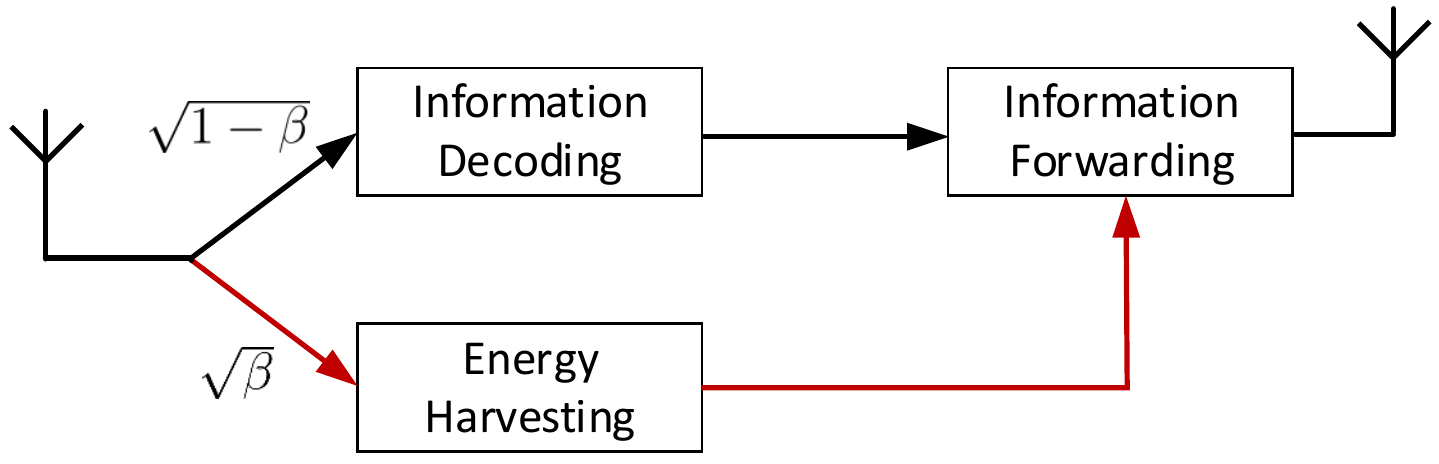}\\
  \caption{The power splitting architecture at user $2$.} \label{swipt}
\end{figure}

According to the NOMA principle, successive interference cancellation (SIC) is performed at user $2$. Particularly, user $2$ first decodes the message to user $1$ (i.e. $x_1$) and then subtracts this message from its observation to decode its own information.

Therefore, by letting ${\bf h}_2 \triangleq \frac{ \sqrt{P_s}}{\sigma_2} \tilde{{\bf h}}_2$, its $\SINR$ 
reads
\begin{align}
\SINR_{2,x_1}^{(1)} &= \frac{ (1-\beta) |{\bf h}_2^{H} {\bf  w}_1|^2} { (1-\beta) |{\bf h}_2^{H} {\bf  w}_2|^2 + 1},
\end{align}
which
should be no less than the target SINR of user 1, denoted by $\gamma_1$, such that the message $x_1$ can be successfully decoded at user $2$.
Hence, we have the constraint
\begin{align}
\frac{ (1-\beta) |{\bf h}_2^{H} {\bf w}_1|^2}{ (1-\beta) |{\bf h}_2^{H} {\bf  w}_2|^2 + 1} \ge \gamma_1.
\end{align}
Then user $2$ subtracts $x_1$ from $y_2^{(1)}$
to further decode its own message $x_2$. The corresponding $\SNR$ can  be expressed as
\begin{align}
\SNR_{2,x_2}^{(1)} &=  (1-\beta) |{\bf h}_2^{H} {\bf  w}_2|^2.
\end{align}

The energy harvested by user $2$ can be modelled as \cite{Shi-2016-tsp}
\begin{align}
E = P_s \beta \left( |\tilde{{\bf h}}_2^{H} {\bf  w}_1|^2 + |\tilde{{\bf h}}_2^{H} {\bf w}_2|^2 \right) \tau,
\end{align}
where $\tau$ is the transmission time fraction for the first stage and we assume that the two stages have the same transmission duration, i.e., $\tau=\frac{1}{2}$.
Assume that the harvested energy is used for information forwarding only, while the energy for maintaining circuit, signal processing, etc., is neglected \cite{Liu-2016}.  Thus, the transmit power at user $2$ is shown to be
\begin{align}
P_t = \frac{E}{1-\tau} = P_s \beta \left( |\tilde{{\bf h}}_2^{H} {\bf w}_1|^2 + |\tilde{{\bf h}}_2^{H} {\bf w}_2|^2 \right).
\end{align}~\vspace{-5mm}

\subsection{Cooperative Transmission Stage}
In this stage, user $2$ forwards $x_1$ to user $1$ by the harvested energy. Therefore, user $1$ observes
\begin{align}
y_1^{(2)} = \sqrt{P_t} g_1 x_1 + z_1^{(2)},
\end{align}
where $g_1\in\mathbb{C}$ is the channel coefficient from user $2$ to user $1$, and $z_1^{(2)}\sim\mathcal{CN}(0,\sigma_1^2)$ is the AWGN at user $1$. Thus, the $\SNR$ to detect $x_1$ is given by
\begin{align}\label{snr2a}
\SNR_{1,x_1}^{(2)} =
 \beta g \left( |{\bf h}_2^{H} {\bf  w}_1|^2 + |{\bf h}_2^{H} {\bf w}_2|^2 \right),
\end{align}
where $g \triangleq |g_1|^2$.

At the end of stage 2, user 1 decodes the message $x_1$ jointly based on the signals received from BS and user $2$ by using MRC. Hence the equivalent SINR at user $1$ can be written as
\begin{align}
\SINR_{1,x_1} &= \SINR_{1,x_1}^{(1)}+\SNR_{1,x_1}^{(2)} \nonumber\\
&= \frac{| {\bf h}_1^{H} {\bf  w}_1|^2}{ |{\bf h}_1^{H} {\bf  \bar{w}}_2|^2 \!+\! 1} \!+\! \beta g \left( |{\bf h}_2^{H} {\bf w}_1|^2 \!+\! |{\bf h}_2^{H} {\bf  w}_2|^2 \right).
\end{align}

The aim of the system design is to maximize the date rate of user $2$, and meanwhile, guarantee the QoS requirement of user $1$. Note that rate maximization of user $2$ is equivalent to maximize its SNR. Then the problem can be formulated as
\begin{subequations}
\begin{align}
\tt P1:~
\max_{ \beta ,{\bf w}_1,{\bf w}_2}~ & (1-\beta) |{\bf h}_2^{H} {\bf w}_2|^2\label{miso_p1.1} \\
\st \quad & \frac{(1-\beta) |{\bf h}_2^{H} {\bf w}_1|^2}{(1-\beta) |{\bf h}_2^{H} {\bf w}_2|^2 + 1} \ge \gamma_1 \label{miso_p1.2},\\
& \frac{| {\bf h}_1^{H} {\bf w}_1|^2}{|{\bf h}_1^{H} {\bf w}_2|^2 + 1} \notag\\
&~~~+ \beta g \left( |{\bf h}_2^{H} {\bf w}_1|^2 + |{\bf h}_2^{H} {\bf  w}_2|^2 \right) \ge \gamma_1 \label{miso_p1.3},\\
& 0 \le \|{\bf w}_1\|^2 + \|{\bf w}_2\|^2 \le 1, \label{miso_p1.4}\\
& 0 \le \beta \le1. \label{miso_p1.5}
\end{align}
\end{subequations}
In {\tt P1}, the objective \eqref{miso_p1.1} is to maximize the receive SNR of user $2$, which is thus called the best-effort user. Constraints \eqref{miso_p1.2} and \eqref{miso_p1.3} indicate the received SINR to decode $x_1$ should be no less than the target SINR, $\gamma_1$. Specifically, \eqref{miso_p1.2} is to ensure $x_1$ can be successfully decoded at user $2$ and \eqref{miso_p1.3} is to guarantee the QoS requirement of user $1$. The power constraint of the transmit beamformers is represented by \eqref{miso_p1.4} and the PS ratio constraint is characterized by \eqref{miso_p1.5}.

Here we remark that the nonconvex problem {\tt P1} is challenging to solve, mainly due to the coupling between the PS ratio and the quadratic terms of the beamformers in the objective function and constraints. In the following, we will first reformulate {\tt P1} by using the SDR technique and then approximately solve the reformulated problem with an SCA-based iterative algorithm.

\section{Suboptimal Beamforming Design and Power Splitting Control}
In this section, the noncovex problem {\tt P1} is firstly transformed into an equivalent problem with the well-known SDR technique \cite{Luo-2010}.
Secondly, some reformulations are conducted to transform the nonconvex terms in the obtained problem into tractable ones.
Finally, an SCA-based iterative algorithm is proposed to approximately solve the reformulated problem.~\vspace{-1mm}

\subsection{Reformulation of {\tt P1} with SDR}
We define two positive semidefinite (PSD) matrices ${\bf W}_1$  and ${\bf W}_2$ as
\begin{align}
{\bf W}_i = {\bf w}_i{\bf w}_i^H \succeq {\bf 0},\quad i = 1,2,~\vspace{-2ex}
\end{align}
and by adopting the SDR technique, \textbf{P1} can be relaxed as
\begin{subequations}
\begin{align}
\tt P2:
\max_{\beta,{\bf W}_1, {\bf W}_2}~ &~(1-\beta)\Tr({\bf H}_2 {\bf W}_2) \label{p2.1} \\
\st~~~ & \frac{(1-\beta) \Tr({\bf H}_2 {\bf W}_1)}{(1-\beta) \Tr({\bf H}_2 {\bf W}_2) + 1} \ge \gamma_1 \label{p2.2},\\
& \frac{\Tr({\bf H}_1 {\bf W}_1)}{\Tr({\bf H}_1 {\bf W}_2) + 1} \notag\\&~~~+ \beta g \Tr\left({\bf H}_2\left({\bf W}_1 + {\bf W}_2\right)\right) \ge \gamma_1 \label{p2.3},\\
& 0 \!\le\! \Tr({\bf  W}_1 \!+\! {\bf W}_2) \le 1, {\bf W}_1,{\bf W}_2 \succeq {\bf 0}, \label{p2.4}\\
& 0 \le \beta \le 1,\label{p2.5}
\end{align}
\end{subequations}
where
${\bf H}_i \triangleq {\bf h}_i{\bf h}_i^H$  for $i = 1,2$.
The constraint \eqref{p2.2} can be rewritten as a convex one, which is given by
\begin{align}
\Tr({\bf H}_2 {\bf  W}_1) - \gamma_1 \Tr({\bf H}_2 {\bf W}_2) \ge \frac{\gamma_1}{1-\beta}. \label{2.1.1}~
\end{align}
Note that the equivalence of {\tt P1} and {\tt P2} cannot be guaranteed due to the nonconvex rank-one constraints are dropped in {\tt P2}, which motivates us to prove the equivalence between the two problems.
Since the objective function of {\tt P2} and the constraints \eqref{p2.3} and \eqref{2.1.1} are all nonlinear expressions,  the conclusion in \cite{Huang-2010} cannot be applied directly to prove the rank-one optimality.
Fortunately, we still have the following proposition.~\vspace{-2mm}

\begin{Proposition}\label{rank_one}
{\tt P2} always has an optimal solution $({\bf W}_1^*,{\bf W}_2^*)$ with $\rank({\bf W}_1^*) = 1$ and $\rank({\bf W}_2^*) \leq 1$, whenever it is feasible.
\end{Proposition}\vspace{-1mm}

\begin{proof}
Firstly, by introducing an auxiliary variable $x\ge 0$ to the constraint \eqref{p2.3}, it can be equivalently rewritten as
\begin{subequations}\label{Eq:16}
\begin{align}
& \Tr({\bf H}_1 {\bf W}_1) \ge  x \Tr({\bf H}_1 {\bf W}_2) + x, \label{Eq:16b}\\
& \beta g \Tr\left({\bf H}_2\left({\bf W}_1 + {\bf W}_2\right)\right) \ge \gamma_1 - x.\label{Eq:16a}
\end{align}
\end{subequations}
Then, for any given $\beta$ and $x$, {\tt P2} is degraded into the following linear SDP problem
\begin{subequations}
\begin{align}
\tt P3:
\max_{{\bf W}_1, {\bf W}_2} &~ (1-\beta)\Tr({\bf H}_2 {\bf W}_2)  \\
\st~ & \Tr({\bf H}_2 {\bf W}_1)- \gamma_1 \Tr({\bf H}_2 {\bf W}_2) \ge \frac{\gamma_1}{1\!-\!\beta}, \label{rank1.2}\\
& \Tr({\bf H}_1 {\bf W}_1) - x \Tr({\bf H}_1 {\bf W}_2) \ge x \label{rank1.2}\\
& \Tr({\bf H}_2{\bf W}_1) + \Tr({\bf H}_2{\bf W}_2) \ge \frac{\gamma_1 - x}{\beta g} ,\\
& \Tr({\bf  W}_1) + \Tr({\bf W}_2) \le 1,\\
& {\bf W}_1\succeq {\bf 0}, {\bf W}_2\succeq {\bf 0}.
\end{align}
\end{subequations}
Assume that {\tt P3} is feasible, and thus it is also dual feasible. Then, according to Theorem 3.2 in \cite{Huang-2010}, we conclude that {\tt P3} always has an optimal solution (${\bf W}_1^*$, ${\bf W}_2^*$) and satisfy
\begin{align}
\rank^2({\bf W}_1^*) + \rank^2({\bf W}_2^*) \le 4.
\end{align}
From \eqref{rank1.2}, we have ${\bf W}_1^* \succeq {\bf 0}$ and ${\bf W}_1^* \neq {\bf 0}$ if {\tt P3} is feasible, thus $\rank({\bf W}_1^*) \ge 1$, and then $\rank({\bf W}_2^*) \le 1$.
\begin{itemize}
\item If $\rank({\bf W}_2^*) = 1$, then we have $\rank({\bf W}_1^*) = \rank({\bf W}_2^*) = 1$.
\item If $\rank({\bf W}_2^*) = 0$, then ${\bf W}_2^* = 0$ and {\tt P3} is equivalent to
\begin{subequations}
\begin{align}
\max_{{\bf W}_1\succeq {\bf 0}}~ & ~0  \\
\st~ & \Tr({\bf H}_2 {\bf W}_1) \ge \max \left\{\frac{\gamma_1}{(1-\beta)}, \frac{\gamma_1 - x}{\beta g} \right\}, \\
& \Tr({\bf H}_1 {\bf W}_1) \ge x,\\
& \Tr({\bf  W}_1) \le 1.
\end{align}
\end{subequations}
Again, by Theorem 3.2 in \cite{Huang-2010}, we can conclude that there exits an optimal solution ${\bf W}_1^*$ which satisfies $\rank^2({\bf W}_1^*) \le 3$, and thus $\rank({\bf W}_1^*) = 1$.
\end{itemize}
Moreover, the optimal $\beta$ and $x$ in {\tt P3} can be found through the two-dimensional exhaustive search method. Thus the rank-one optimality of {\tt P2}, i.e., the SDR tightness, is proved.\hfill $\blacksquare$
\end{proof} \vspace{-2mm}
\begin{rem}
  The special case with ${\bf W}_2 = {\bf 0}$ indicates that the system can just satisfy the QoS requirement of user 1. When the QoS requirement of user 1 is high or its channel condition is weak, it tends to sacrifice the data rate of user $2$ to guarantee its own QoS requirement.
\end{rem}
\begin{rem}
The optimal solution to {\tt P3} is not always guaranteed to be rank-one. Hence,  if $\rank({\bf W}_1^*)=\rank({\bf W}_2^*)=1$, then the optimal beamforming vectors ${\bf w}_1^*$ and ${\bf w}_2^*$ to {\tt P1} can be obtained from ${\bf W}_1^*$ and ${\bf W}_2^*$ by eigen-decomposition, otherwise, Gaussian randomization procedure can be utilized to yield a suboptimal solution \cite{Luo-2010}.
\end{rem}

Due to the high complexity of the two-dimensional exhaustive search method, it motivates us to develop a more efficient algorithm to solve {\tt P2}.
Notice that, although all beamformer related quadratic terms in {\tt P1} are linearized in the semidefinite cone,
{\tt P2} is still mathematically intractable, as explained in the following.
By verifying the Hessian of the objective function, it can be readily verified that this is a nonconvex function (in fact, a sloping saddle surface function). Furthermore, \eqref{p2.3} is also nonconvex due to the coupling  among the variables $\beta$, ${\bf{W}_1}$ and ${\bf{W}_2}$.
In the next two subsections, we will first reformulate the nonconvex 
expressions \eqref{p2.1} and \eqref{p2.3} in {\tt P2} and then develop an SCA-based algorithm to efficiently obtain a stationary point of {\tt P2}.


\subsection{Reformulation of Nonconvex Constraints in {\tt P2}}
By epigraph reformulation \cite{Boyd-2009}, the objective function \eqref{p2.1} can be equivalently rewritten as
\begin{subequations}
\begin{align}
\max_{u,v,\beta,{\bf W}_2} \quad & u \label{obj0}\\
\st ~\quad & v^2 \ge u, \label{obj1}\\
&\begin{bmatrix}
     (1-\beta)   & v \\
     v & \Tr\left({\bf H}_2{\bf  W}_2\right)
\end{bmatrix}
\succeq 0, \label{obj2}
\end{align}
\end{subequations}
which consists of a linear objective function \eqref{obj0}, a nonconvex quadratic inequality constraint \eqref{obj1} and a convex linear matrix inequality (LMI) constraint \eqref{obj2}.

{Recall that \eqref{p2.3} is equivalent to \eqref{Eq:16}, where}
\eqref{Eq:16b} is a nonconvex constraint and it will be dealt with in the next subsection, while
\eqref{Eq:16a} can be transformed into a nonconvex quadratic constraint and an LMI constraint as below
\begin{subequations}
\begin{align}
t^2 &\ge \gamma_1 - x, \label{2.2.2.1}\\
\begin{bmatrix}
     g \beta  & t \\
     t & \Tr\left({\bf H}_2\left({\bf W}_1 + {\bf W}_2\right)\right)
\end{bmatrix}
&\succeq 0. \label{2.2.2.2}
\end{align}
\end{subequations}

Until now, {\tt P2} becomes
\begin{subequations}
\begin{align}
{\tt P4}:
\max_{u,v,t,x,{\bf W}_1, {\bf W}_2,\beta} & \quad u \label{p3.1} \\
\st~~\quad
&\!\!\!\!\!\! v^2 \ge u,  \label{p3.2}\\
&\!\!\!\!\!\! t^2 \ge \gamma_1 - x,  \label{p3.3}\\
&\!\!\!\!\!\! \Tr({\bf H}_1 {\bf  W}_1) - x \ge \Tr({\bf H}_1 {\bf  W}_2 )x, \label{p3.4}\\
&\!\!\!\!\!\! \eqref{p2.4}, \eqref{p2.5}, \eqref{2.1.1}, \eqref{obj2}~{\rm and}~\eqref{2.2.2.2}.\nonumber
\end{align}
\end{subequations}
Compared to {\tt P2}, {\tt P4} explicitly reveals the fundamental difficulties of {\tt P1} that lie in the nonconvex constraints \eqref{p3.2}-\eqref{p3.4}.
In the following, we will use an iterative approach to approximate this problem, in which constraints \eqref{p3.2} and \eqref{p3.3} are replaced by their first-order Taylor expansions and constraint \eqref{p3.4} is approximated by using the arithmetic-geometric mean (AGM) inequality.


\subsection{SCA-based Algorithm for {\tt P4}}

The key idea of the SCA method is to iteratively approximate the nonconvex problem by convex ones \cite{Li-2013-tsp}. Performing the first-order Taylor approximation, the nonconvex constraints \eqref{p3.2} and \eqref{p3.3} can be approximated as below
\begin{subequations}
\begin{align}
 2 v^{(n)} v - (v^{(n)})^2 &\ge u, \label{p3.2.1}\\
 2 t^{(n)} t - (t^{(n)})^2 &\ge \gamma_1 - x, \label{p3.3.1}
\end{align}
\end{subequations}
where $v^{(n)}$ and $t^{(n)}$ denote the value of variable $v$ and $t$ at the $n$-th iteration. 

For constraint \eqref{p3.4}, the AGM inequality can be used to yield approximate constraint. For any nonnegative variables $x,y,z$, the AGM inequality-based approximation of the nonconvex expression $xy \le z$, which has the same form as \eqref{p3.4}, can be described as
\begin{align}
2xy \le (ax)^2 + (y/a)^2 \le 2z,
\end{align}
where the first equality holds if and only if $a = \sqrt{y/x}$.
Therefore, the constraint \eqref{p3.4} can be approximated by the following convex constraint
\begin{align}
&(a^{(n)}x)^2 \!+\! (\Tr({\bf H}_1 {\bf W}_2 )/a^{(n)})^2 \!\le\! 2\Tr({\bf H}_1 {\bf W}_1) \!-\! 2x,
\end{align}
where $a^{(n)}$ is the value of $a$ at the $n$-th iteration and can be updated by
\begin{align}\label{agm}
a^{(n)} &= \sqrt{\left(\Tr({\bf H}_1 {\bf W}_2 )\right)^{(n-1)}/x^{(n-1)}}.
\end{align}

Hence, the problem that needs to be solved during the $n$-th iteration is given by
\begin{subequations}
\begin{align}
{\tt P5}:
\max_{u,v,t,x,\beta,{\bf W}_1, {\bf W}_2} \quad &u \label{p4.1}\\
\st~~~~~~~ ~
&\!\!\!\!\!\!\!\! 2 v^{(n)} v - (v^{(n)})^2 \ge u, \label{p4.2}\\
&\!\!\!\!\!\!\!\! 2 t^{(n)} t - (t^{(n)})^2 \ge \gamma_1 - x,   \label{p4.3}\\
&\!\!\!\!\!\!\!\! (a^{(n)}x)^2 +(\Tr({\bf H}_1 {\bf W}_2 )/a^{(n)})^2\notag\\&~~~~~\le 2\Tr({\bf H}_1 {\bf W}_1) - 2x, \label{p4.5}\\
&\!\!\!\!\!\!\!\! \eqref{p2.4}, \eqref{p2.5}, \eqref{2.1.1}, \eqref{obj2}~{\rm and}~ \eqref{2.2.2.2}, \nonumber
\end{align}
\end{subequations}
which is a convex problem. Thus it can be efficiently handled with the off-the-shelf convex solver, e.g., \texttt{CVX} \cite{cvx}.

The SCA-based algorithm is outlined in Algorithm \ref{Alg1}. The main part of this algorithm is an iterative procedure, which starts by solving {\tt P4} to get an approximated solution with the initial value of $u_{0}$, $v_{0}$, $t_{0}$, $a_{0}$ and update the value of $u^{(n)}$, $v^{(n)}$, $t^{(n)}$, $a^{(n)}$ after each iteration according to the obtained solution. The iterative procedure will repeat until the gap of the objective function between two successive iterations is below a threshold $\epsilon_1$. In fact, we can draw the following proposition.

\begin{Proposition}\label{convergence}
The proposed algorithm can continuously decrease the rate gap between two successive iterations and guarantee the generated rate sequence converges to at least a stationary point whenever {\tt P4} is feasible.
\end{Proposition}

\begin{proof}
The proof of Proposition \ref{convergence} is similar to \cite{Li-2013-tsp}, and we omit it here.
\end{proof}

\begin{algorithm}[!t]\small
\caption{~SCA Method for Solving \tt P4 }\label{Alg1}
\begin{algorithmic}[1]
 \STATE {{\bf Initialization:} Set $n=0$, $u_0 = -\infty$, $v_0 = 1$, $t_0 = 1$, $a_0 = 1$, $\Delta = 1$ and the tolerance $\epsilon = 10^{-4}$.}
 \WHILE{$\Delta \ge \epsilon$}
 \STATE {Update $\beta^{(n)}$, ${\bf W}_1^{(n)}$ and ${\bf W}_2^{(n)}$ by solving {\tt P5}.}\\
 \STATE {Update $a^{(n)}$ based on \eqref{agm}.}\\
 \STATE {Update $\Delta = |u^{(n)} - u^{(n - 1)}|$.}\\
 \STATE {Set $n \leftarrow n +1$.}
 \ENDWHILE
 \ENSURE {$\beta^{(n)}$, ${\bf W}_{1}^{(n)}$ and ${\bf W}_{2}^{(n)}$.}
\end{algorithmic}
\end{algorithm}

\begin{rem}
It is worth to pointing out that the objective function (14a) can also be alternatively reformulated as
   \begin{subequations}
   \begin{align}
   \!\!\!\!\!\!\!\max ~&~t\\
   \st ~&\left(\theta(1-\beta) + \Tr({\bf H}_2{\bf W}_2)\right)^2 - \theta^2(1-\beta)^2 \notag \\
   &~~~~~~~~~~~~~~~~~~~~~~~~~~~ - \left(\Tr({\bf H}_2{\bf W}_2)\right)^2 \ge 2\theta t,
   \end{align}\label{DoQ}
   \end{subequations}
\!\!where the preconditioning parameter $\theta\in\Rset^+$ is chosen to balance the terms $1-\beta$ and $\Tr({\bf H}_2 {\bf W}_2)$ such that potential conditioning of the problem can be improved with more favorable properties for the solution.
   The relevant constraints can also be expressed with the form of difference of quadratic functions. Then the idea of SCA can be implemented over \eqref{DoQ}, and it can also converge to a stationary point whenever {\tt P4} is feasible \cite{Li-2013-tsp,messiam-2013} .

\end{rem}

In the previous sections, we focus on the system design of the cooperative SWIPT NOMA scheme in MISO cases. However, from the perspective of practical applications, e.g.,  in the IoT scenario, the devices are typically equipped with a single antenna.  Motivated by this, we will further consider the cooperative SWIPT NOMA scheme under the SISO cases. The resultant problem is a special case of {\tt P1} with $N_t = 1$, hence it can be effectively solved by Algorithm \ref{Alg1}. However, only a stationary point can be attained. In the following section, we will present an iterative algorithm, through which the \textit{global optimality} of the obtained solution can be guaranteed. In addition, the optimum can be written as a \textit{closed-form} expression in each iteration.
\vspace{-2mm}

\section{Optimal Transmission Protocol Design for SISO Cases}
In this section, we consider the cooperative SWIPT NOMA transmission protocol design in SISO cases, in which all the nodes (i.e., the BS, user $1$ and user $2$) are equipped with a single antenna.
\vspace{-2mm}

\subsection{Problem Formulation in SISO Cases}
In SISO cases, the beamforming design in {\tt P1} is degraded  into a power control problem. The beamforming vector ${\bf w}_1$ is replaced by a power control variable $\alpha$, representing the power fraction for message $x_1$. Then, the optimization problem in SISO cases can be formulated as
\begin{subequations}
\begin{align}
\tt P6: \quad
\max_{\alpha,\,\beta} \quad &(1-\alpha) (1-\beta)h_2    \label{P5.1}\\
\st \quad & \frac{(1-\beta)\alpha h_2}{(1-\beta)(1-\alpha)h_2 + 1} \ge \gamma_1, \label{P5.2}\\
& \frac{\alpha h_1}{(1-\alpha) h_1 + 1} + \beta h_2 g \ge \gamma_1, \label{P5.3}\\
& 0 \le \alpha \le 1, ~0 \le \beta \le 1, \label{P5.5}
\end{align}
\end{subequations}
where $h_1\in \Rset^+$, $h_2\in\Rset^+$ and $g\in\Rset^+$ are normalized channel gains, respectively.

As mentioned before, {\tt P6} can be effectively solved with Algorithm \ref{Alg1} to yield a stationary point. To obtain the global optimum, in the next subsection, we will present a golden section search (GSS) based iterative algorithm, by which the global optimum of {\tt P6} admits a \textit{semiclosed-form} expression.
\vspace{-2mm}


\subsection{Global Optimal Solution to {\tt P6}}
Recall that the PS ratio $\beta$ represents the amount of signal power for EH of user 2. So, if all the power received by user $2$ is used for EH, i.e., $\beta = 1$, then the user $2$ has a zero data rate and cannot help user $1$.
Correspondingly, if $\beta = 0$, namely, no EH at user 2 and all its received power is used for information decoding, however, it has also no energy to help user 1 and then the QoS requirement of user $1$ may not be guaranteed.
So, the optimal value of $\beta$ for {\tt P6} lies in $[0,~1)$. In fact, the feasible set of $\beta$ is characterized by the following lemma.\vspace{-2mm}
\begin{Proposition} \label{feasibility_set}
The feasible set of $\beta$ is $\left[\beta_{\rm min},\beta_{\rm max}\right]$, where $\beta_{\rm min} = \frac{(\gamma_1 - h_1)^+}{h_2g}$ and $\beta_{\rm max} = 1-\frac{\gamma_1}{h_2}$.
\end{Proposition}

\begin{proof}
\eqref{P5.2} implies that the message of user $1$ can be successfully decoded by user $2$, while \eqref{P5.3} guarantees the QoS target of user $1$. Hence, for any $\alpha \in [0,1]$, \eqref{P5.2} and \eqref{P5.3} require that
\begin{align}
\frac{\gamma_1 - \frac{\alpha h_1}{(1-\alpha)h_1 + 1}}{h_2 g} \le \beta \le 1 - \frac{\gamma_1}{h_2 - (\gamma_1+1)(1-\alpha)h_2}. \label{beta1}
\end{align}

Both \eqref{P5.2} and \eqref{P5.3} prefer a larger $\alpha$, consequently, taking $\alpha = 1$, \eqref{beta1} becomes
\begin{align}
\frac{\gamma_1 - h_1}{h_2g} \le \beta \le 1-\frac{\gamma_1}{h_2}. \label{beta2}
\end{align}
Recall that  $0 \le \beta < 1$, the feasible set of $\beta$ is thus given by $\left[\frac{(\gamma_1 - h_1)^+}{h_2g}, 1-\frac{\gamma_1}{h_2} \right]$,  where $(x)^+ = \max(x,0)$.  This completes the proof.
\end{proof}

From the proof of Lemma {\ref{feasibility_set}}, some interesting observations can be obtained.
\begin{rem} \label{set}
The achievable data rate for user $2$, denoted by $R_2$, is definitely $0$ when $\beta = \beta_{\max}$, and $R_2 \ge 0$ if $\beta = \beta_{\min}$. Specifically, $R_2 > 0$ if $h_1>\gamma_1 $ and $0$ otherwise, which indicates that the user 2 can achieve a positive data rate only when the normalized channel gain, $h_1$, is large enough to support the required \SINR, $\gamma_1$.
\end{rem}

For the convenience of illustration, we transform {\tt P6} into a bilevel programming problem, with the upper-level variable being $\beta$, which is given by
\begin{subequations}
\begin{align}
\max_{\beta \in [\beta_{\rm min},\beta_{\rm max}]} h(\beta) \!\triangleq\! &\max_{0 \le \alpha \le 1}~(1-\alpha)(1-\beta)h_2 \label{p7.1}\\
&\st~\frac{(1-\beta)\alpha h_2}{(1-\beta)(1-\alpha)h_2 + 1} \ge \gamma_1, \label{p7.2}\\
&~~~~~~\frac{\alpha h_1}{(1-\alpha) h_1 + 1} \!+\! \beta h_2 g \ge \gamma_1, \label{p7.3}
\end{align}
\end{subequations}
where $h(\beta)$ is the inner optimization problem with respect to $\alpha$. Now, we can observe that, for any given $\beta\in \left[\beta_{\rm min},\beta_{\rm max}\right]$, \eqref{p7.2} and \eqref{p7.3} can be rewritten respectively as
\begin{subequations}
\begin{align}
\alpha &\ge \frac{\gamma_1 ((1-\beta)h_2+1)}{(1+\gamma_1)(1-\beta)h_2},\\
(\gamma_1+1-\beta h_2 g)h_1\alpha &\ge (\gamma_1-\beta h_2 g)(h_1+1).
\end{align}
\end{subequations}
So, the optimal $\alpha$ is given by the following \textit{closed-form} expression:
\begin{align}
\alpha(\beta) =
\begin{cases}
{\rm min} \{A,1\}, \quad \quad &{\rm if}~ \beta \ge \frac{\gamma_1}{h_2g},\\
{\rm min} \{ {\rm max}\{A,B\},1\}, &{\rm otherwise},\\
\end{cases} \label{alpha}
\end{align}
in which $A = \frac{\gamma_1 ((1-\beta)h_2+1)}{(1+\gamma_1)(1-\beta)h_2}$, $B = \frac{(\gamma_1-\beta h_2 g)(h_1+1)}{(\gamma_1+1-\beta h_2 g)h_1}$.

From \eqref{alpha}, we can get the following useful observation.
\begin{rem} \label{rem2}
If $\beta \ge \gamma_1/(h_2 g)$, then \eqref{p7.2} must hold with equality, otherwise, \eqref{p7.3} holds with equality.
\end{rem}

With given $\beta$, the optimal $\alpha$ and thus $h(\beta)$ is determined through \eqref{alpha}. Hence, the key is to get the optimal $\beta$ in an efficient way. To this end, we first prove that $h(\beta)$ possesses the following property, which can be used to derive an iterative algorithm to obtain the optimal $\beta$.

\begin{Proposition} \label{lem_unimodal}
$h(\beta)$ is strictly unimodal with respect to $\beta \in [\beta_{\rm min},\beta_{\rm max}]$.
\end{Proposition}

\begin{proof}
The sketch of the proof for the proposition is provided in the following first. If $h(\beta)$ is strictly unimodal in $\beta$, then it is in fact restricted to being strictly monotonically increasing, or strictly monotonically decreasing, or strictly monotonically increasing until $\hat{\beta} \in [\beta_{\min}, \beta_{\max}]$ and then strictly monotonically decreasing.
In the following, we will show that $h(\beta)$ is either monotonically decreasing or increasing first then decreasing in $\beta$.

With some algebraic manipulations, we can find that 
$h(\beta)$ prefers a smaller $\beta$ and a smaller $\alpha$;  \eqref{p7.2} prefers a smaller $\beta$ and a larger $\alpha$, while  \eqref{p7.3} prefers a larger $\beta$ and a larger $\alpha$. 
\begin{lem}\label{lem1}
For any given $\beta \in [\beta_{\rm min},\beta_{\rm max}]$, the optimal $\alpha$ will make at least one of the two constraints \eqref{p7.2} and \eqref{p7.3} holds the equality.
\end{lem}
\begin{proof}
See Appendix \ref{app1}.
\end{proof}

Based on Lemma  \ref{lem1}, we can conduct the proof in the following two cases.

\textit{\underline{Case 1}}: \eqref{p7.2} holds the equality with $\beta = \beta_{\min}$.

Assuming $\beta = \beta_{\rm min}$, with the increasing of $\beta$, $\alpha$ should also be increased such that both \eqref{p7.2} and \eqref{p7.3} can be satisfied. However, $h(\beta)$ will be strictly decreased. Recall that if $\beta$ increases to $\beta_{\rm max}$, $\alpha$ will reach to $1$, and then any larger $\beta$ will make {\tt P6} infeasible. In summary, $h(\beta)$ is strictly monotonically decreasing with the increase of $\beta$ in this case.

\textit{\underline{Case 2}}: \eqref{p7.3} holds the equality with $\beta = \beta_{\min}$.

Assuming $\beta = \beta_{\rm min}$, with the increase of $\beta$, the value of $\alpha$ can be decreased, such that both \eqref{p7.2} and \eqref{p7.3} can be satisfied. However, the value of the left side of \eqref{p7.2} continuously decreases with the increase of $\beta$ and the decrease of $\alpha$ until the constraint \eqref{p7.2} becomes active, under which, the solution of $\beta$  to \eqref{p7.2} and \eqref{p7.3} is denoted by $\beta_0$.
Moreover, we have $\beta_0 \le \gamma_1/h_2g$ (refer to Remark \ref{rem2}). Then, the same thing happens as in \textit{Case 1}, i.e., $h(\beta)$ is monotonically decreasing when $\beta \in [\beta_0,\beta_{\rm max}]$.

Now let us focus on the monotonicity of $h(\beta)$ when $\beta \in [\beta_{\min},\beta_0]$.
Since $\alpha$ will be decreased with increasing $\beta  \in [\beta_{\min},\beta_0]$,  the monotonicity of $h(\beta)$ can not be verified directly.
Recall that
\eqref{p7.3} holds the equality in this case (refer to Remark \ref{rem2}), so the optimal $\alpha$ can be written as
\begin{align}
\alpha = \min  \left\{B,1\right\}.
\end{align}
Based on the following lemma, we further conclude that it is always true that $\alpha = B$ when $\beta \in [\beta_{\min},\beta_0]$.
\begin{lem} \label{prop2}
For any $\beta \in [\beta_{\rm min},\beta_0]$, it is always true that $0 \le B \le 1$.
\end{lem}
\begin{proof}
See Appendix \ref{app2}
\end{proof}

Define a function $g(\alpha,\beta) = \alpha\beta - \alpha - \beta$, which has the same unimodality as $h(\beta)$.
Inserting $\alpha=B$ yields
\begin{align}
f(\beta)&\triangleq g(B, \beta) = B \beta - B - \beta, \nonumber\\
                            &= \frac{(\gamma_1-\beta h_2 g)(h_1+1)(\beta - 1) }{(\gamma_1-\beta h_2 g+1)h_1}- \beta .
\end{align}
Now the monotonicity of $f(\beta)$ can be verified by its first-order derivative, as given by
\begin{align}
&\!\!\!\!\!f^{\prime}(\beta) = \notag\\
&\!\!\!\!\!\frac{\left(h_2^2 g^2 \beta^2 \!-\! 2h_2g(\gamma_1+1)\beta \!+ \! h_2g \!+ \!\gamma_1^2 \!+\! \gamma_1\right)(h_1^2 \!+\! h_1)}{\left(\gamma_1-\beta h_2 g+1\right)^2 h_1^2} - 1. \label{first-order}
\end{align}
Please refer to Appendix \ref{app3} for the details.

Notice that $\left(\gamma_1-\beta h_2 g+1\right)^2 h_1^2 > 0$,
so it suffices to verify the monotonicity of $f(\beta)$ by the difference between the numerator and denominator of the first term of the right side of \eqref{first-order}. The difference is given by
\begin{align}
\!\!\!\!\!\!\Delta(\beta) \!\triangleq\! h_1\left(h_2 g \beta \!-\! (\gamma_1 \!+\! 1)\right)^2 \!+\! \left(h_2g \!-\! (\gamma_1\!+\!1) \right)(h_1^2 \!+\! h_1). \label{diff}
\end{align}
Please refer to Appendix \ref{app4} for the details.

 Recall that $\beta_0 \le \frac{\gamma_1}{h_2 g}$, so $\Delta(\beta)$ is continuous and strictly monotonically decreasing with $\beta \in \left[\beta_{\rm min},\beta_0\right]$. Thus we have
\begin{itemize}
  \item If $\Delta(\beta_{\min}) < 0$, then $\Delta(\beta) < 0$ and $f(\beta)$ decreases strictly monotonically with $\beta \in \left[\beta_{\rm min},\beta_0\right]$.
  \item If $\Delta(\beta_0) > 0$, then $\Delta(\beta) > 0$ and $f(\beta)$ increases strictly monotonically with $\beta \in \left[\beta_{\rm min},\beta_0\right]$.
  \item If $\Delta(\beta_{\min}) \ge 0$ and $\Delta(\beta_0) \le 0$, then there exists an unique $\hat{\beta} \in [\beta_{\min}~ \beta_0]$ which makes $f(\beta)$ increase  strictly monotonically with $\beta \in [\beta_{\min}, ~\hat{\beta}]$ and decrease  strictly monotonically with $\beta \in [\hat{\beta},~ \beta_0]$.
\end{itemize}
In summary, in Case 2, $h(\beta)$ either strictly monotonically decreases with $\beta \in [\beta_{\min},~\beta_{\max}]$, or strictly monotonically increases with $\beta \in [\beta_{\min},~ \hat{\beta}]$ and then strictly monotonically decreases with $\beta \in [\hat{\beta},~\beta_{\max}]$.

In conclusion, $h(\beta)$ is strictly unimodal with respect to $\beta \in [\beta_{\min}, \beta_{\max}]$ in both Case 1 and 2. This completes the proof.
\end{proof}
\begin{rem}
From above, we conclude that $h(\beta)$ cannot be an  strictly monotonically increasing function. Otherwise, {\tt P6} is feasible if and only if $\beta=\beta_{\max}$, since $h(\beta_{\max}) = 0$ always holds true (refer to Remark \ref{set}).

\end{rem}

As $h(\beta)$ is strictly unimodal with respect to $\beta$, the optimal $\beta$ can be obtained through the GSS method and, for a fixed $\beta$, the optimal $\alpha$ is given by \eqref{alpha}. So the \textit{semiclosed-form} solution of {\tt P6} can be attained through a GSS-based algorithm as shown in Algorithm \ref{alg2} .
Here we also remark that Algorithm \ref{alg2} converges to the unique global optimal solution to {\tt P6} due to its strict unimodality.
\begin{algorithm}\small
	\caption{~GSS-based algorithm for solving {\tt P6}} \label{alg2}
	\begin{algorithmic}[1]
		\STATE {{\bf Initialization:} Set $b_{\min} = 0$, $b_{\max} = 1$, the GSS parameter $a = 0.618$ and the tolerance $\epsilon = 10^{-4}$.}\\
		\WHILE {$b_{\max} - b_{\min} > \epsilon$}
		\STATE {Set $b_1 = b_{\min} + (1-a)(b_{\max} - b_{\min}$) and $b_2 = b_{\min} + a(b_{\max} - b_{\min}$).}
		\STATE {Compute the optimal $\alpha(b_1)$ and $\alpha(b_2)$ based on \eqref{alpha}.}
		\STATE {Compute $h(b_1)$ and $h(b_2)$ based on \eqref{p7.1}.}\\
		\STATE {{\bf if} $h(b_1) > h(b_2)$ {\bf then}}\\
		\STATE {\quad Update $b_{\max} = b_2$.}\\
		\STATE {{\bf else}}\\
		\STATE {\quad Update $b_{\min} = b_1$.}\\
		\STATE {{\bf endif}}\\
		\ENDWHILE
		\ENSURE {$\beta^* = \frac{b_{\max}+b_{\min}}{2}$ and $\alpha(\beta^*)$.}
	\end{algorithmic}
\end{algorithm}
\begin{rem}
Due to the strict unimodality of {\tt P6} with $\beta$, we can draw the fact that the stationary point obtained by Algorithm \ref{Alg1} in SISO cases is exactly the unique global optimal solution whenever {\tt P6} is feasible, which will be verified in the simulation section.
\end{rem}


\section{Numerical Results}
In this section, some simulation results are shown to evaluate the performance of the proposed cooperative SWIPT NOMA protocol.
\subsection{Simulation Setup}
In the simulations, we assume the two users are randomly allocated in a 5-meter $\times$ 6-meter room and the BS is fixed at the edge with a coordinate $(0,2.5 {\rm m})$. The distance-dependent pass loss is modeled by $P_L = 10^{-3}d^{-\alpha}$, in which $d$ denotes the Euclidean distance in meters and $\alpha$ is the path loss exponent. Without loss of generality, we assume the path loss exponents for user $1$ and user $2$ are $\alpha_1 = 4$ and $\alpha_2 = 2$, respectively. Assume that the noise power is $\sigma_1^2 = \sigma_2^2 = \sigma^2 = -90$ dBm and the total system bandwidth is $1$ MHz. By applying the Racian fading channel model, the downlink channels are modeled as
\begin{subequations}
\begin{align}
{\bf \tilde{h}}_{2} &= \sqrt{\frac{K}{1+K}}{\bf h}_2^{\rm LOS} + \sqrt{\frac{1}{1+K}}{\bf h}_2^{\rm NLOS}\\
\tilde{g}_{1} &= \sqrt{\frac{K}{1+K}}g_1^{\rm LOS} + \sqrt{\frac{1}{1+K}}g_1^{\rm NLOS}
\end{align}
\end{subequations}
where $K = 3$ denotes the Rician factor,
 ${\bf h}_2^{\rm LOS}$ and $g_1^{\rm LOS}$ follow the LOS deterministic components,
 and ${\bf \tilde{h}}_1$, ${\bf h}_2^{\rm NLOS}$ and $g_1^{\rm NLOS}$ are the standard Rayleigh fading components with zero mean and unit variance.
All the simulation results are averaged over 1,000 independent channel realizations.

For comparison, we introduce some other transmission strategies, namely the noncooperative NOMA strategy, OMA with dynamic time allocation and OMA with  fixed time allocation, which
are described as follows
\begin{itemize}
  \item In the \textit{noncooperative NOMA strategy}, the system still performs the NOMA strategy. However, there is no SWIPT operation at user 2, i.e.,  the cooperative transmission stage is removed. The BS designs the beamforming vectors to maximize the data rate of user $2$ and, meanwhile, satisfy the QoS requirement of user $1$.
  \item In the \textit{OMA with dynamic time allocation}, the system operates in the TDMA mode, i.e., the BS sends information to user $1$ and user $2$ in different time intervals. Moreover, the time allocation scheme obeys a dynamic manner.
  \item In the \textit{OMA with fixed time allocation}, the system operates in the TDMA mode and the time resource is evenly allocated to user $1$ and user $2$.
\end{itemize}

\subsection{Sum Rate of Users}
\begin{figure}[!tp]
  \centering
  \!\!\!\!\includegraphics[width=0.96\linewidth]{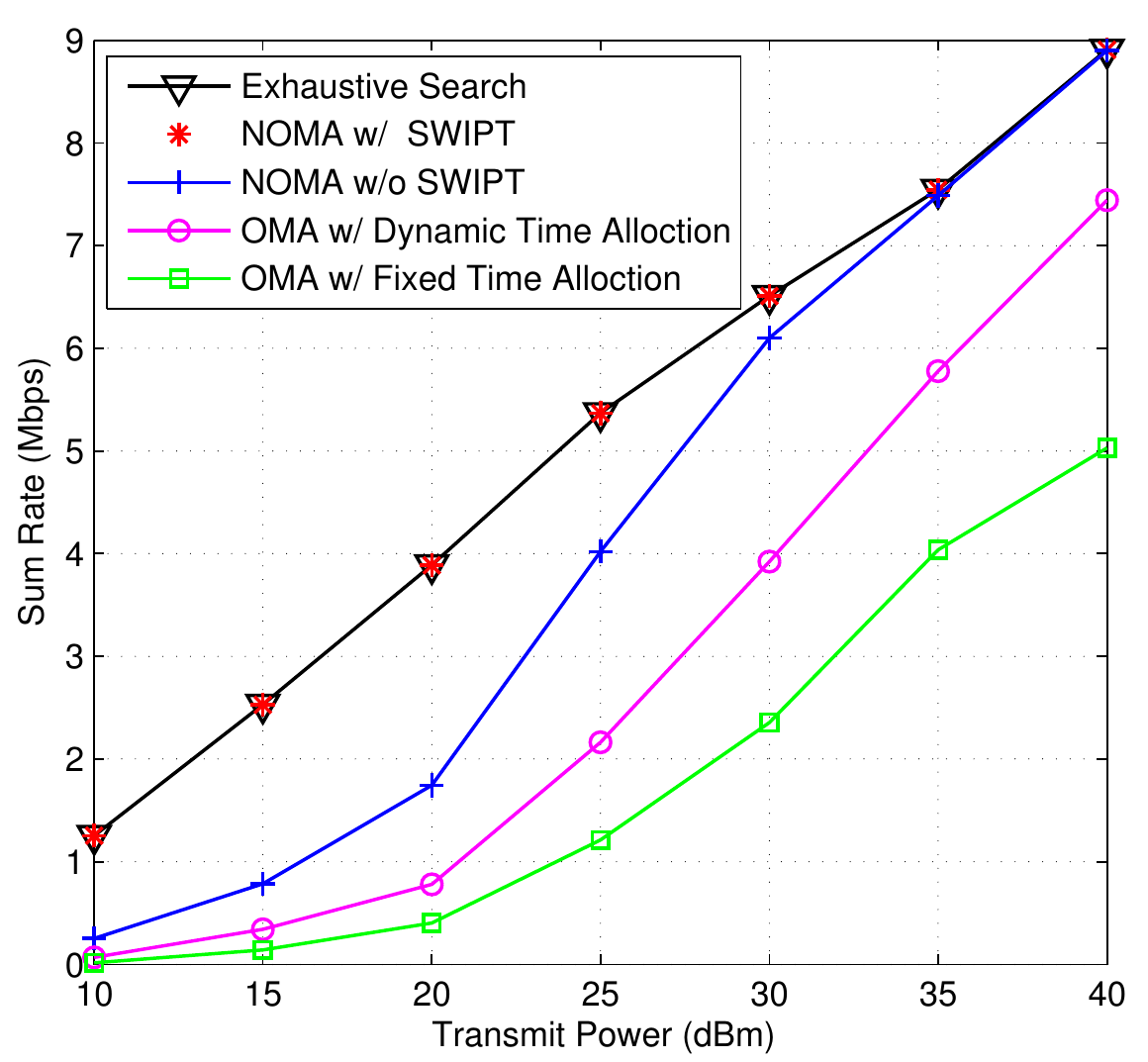}\\
  \caption{Sum rate comparison of different transmission strategies in MISO cases with $N_t = 2$ and $\gamma_1 = 1$.}\label{comp_miso}
\end{figure}

Fig.\ref{comp_miso} depicts the sum rate of two users versus transmission power of the BS in the MISO case.
            One can observe that the performance of the proposed SCA-based algorithm is very close to that of the exhaustive search method. Also, we find that
            the proposed cooperative SWIPT NOMA strategy yields the best performance among all the considered transmission strategies.
            Specifically, the proposed strategy outperforms the noncooperative NOMA scheme in the low power regime and has the same performance in the high power regime. This is because the cooperative transmission provides a higher gain, which is very useful to improve the reception reliability when the system experiences deep fading in the low power regime.
Besides, the noncooperative NOMA strategy achieves a better performance than the OMA strategy, indicating the advantage of NOMA in improving the system SE.
Moreover, the OMA with dynamic time allocation outperforms that with fixed time allocation due to its extra flexibility on resource allocation.
The system performance comparison in SISO cases is shown in Fig. \ref{comp_siso}, it can be seen that the proposed strategy always achieves a better performance than any other considered strategies.

\begin{figure}[!tp]
  \centering
  \includegraphics[width=0.92\linewidth]{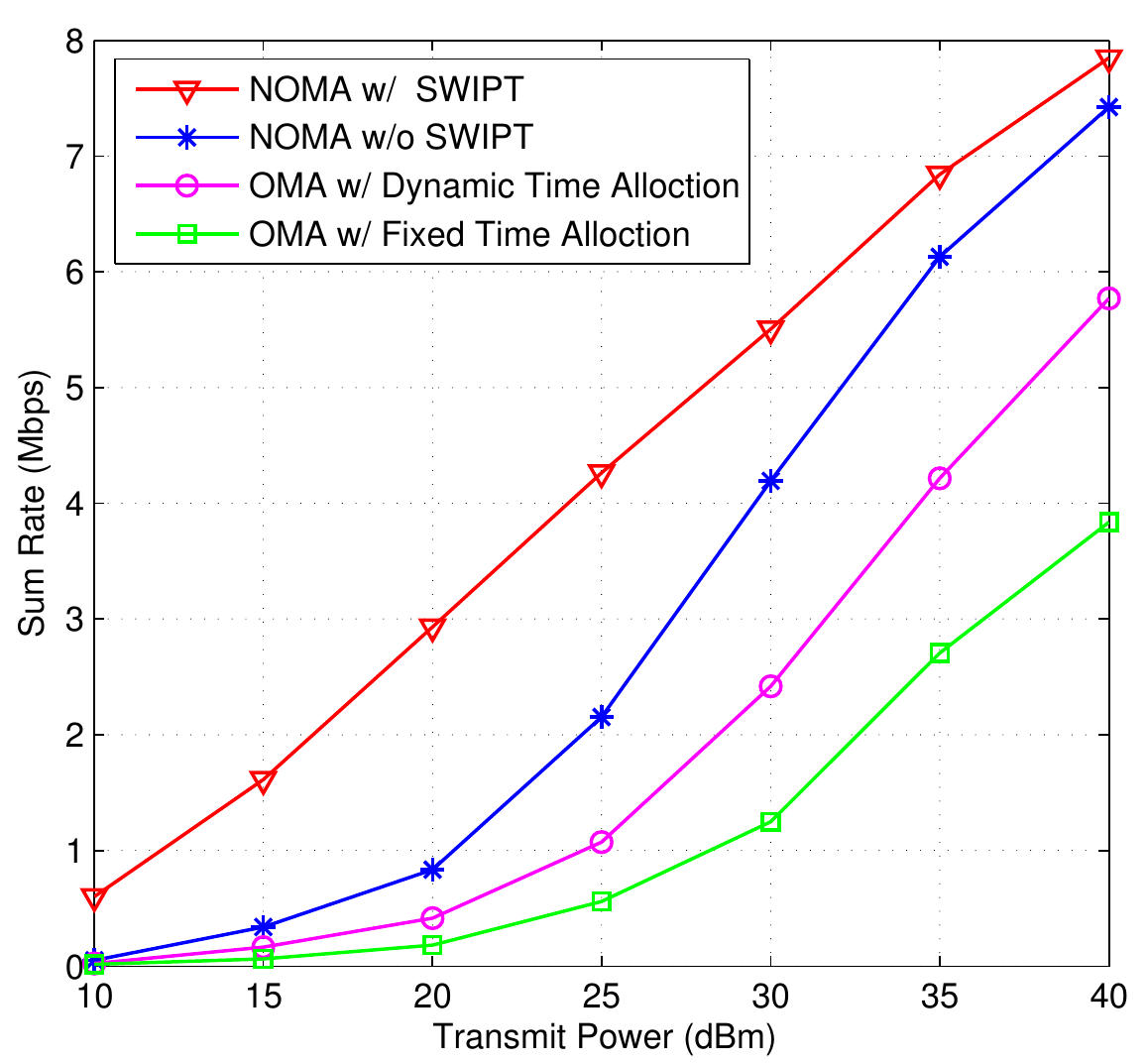}\\
  \caption{Sum rates comparison of different transmission strategies in SISO cases with $\gamma_1 = 1$.}\label{comp_siso}
\end{figure}

\begin{figure}[!tp]
  \!\!\!\!\includegraphics[width=0.96\linewidth]{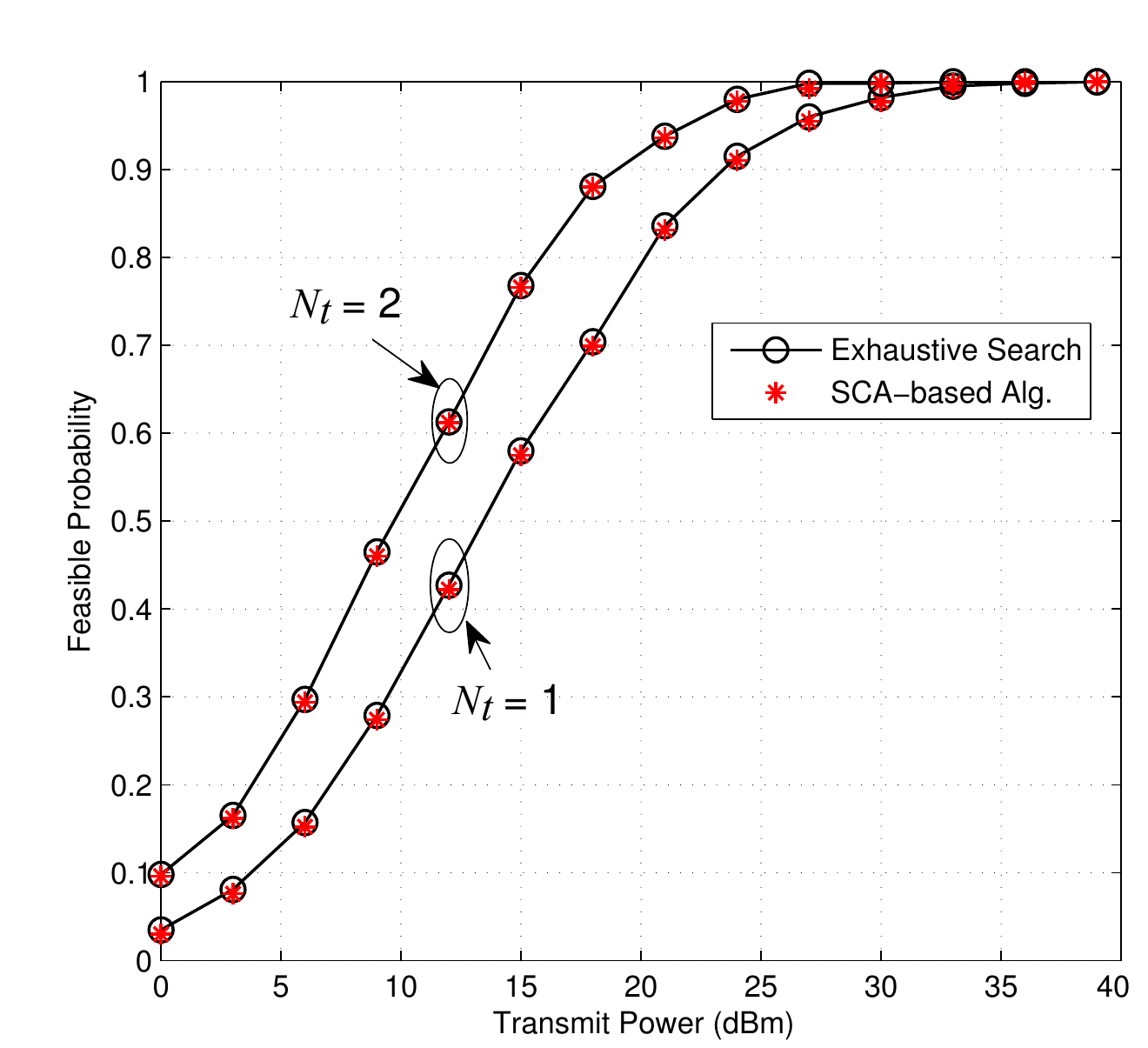}\\
  \vspace{-4.5mm} \caption{Feasible probability comparison between exhaustive search based and proposed algorithms in MISO and SISO cases with $\gamma_1 = 1$.} \label{feasible_probability_power}
\end{figure}

Fig. \ref{feasible_probability_power} depicts the feasible probability versus transmit power of the BS with different algorithms in MISO cases and SISO cases. It can be seen that the feasible probability increases with the transmit power. Besides, the feasible probability of the proposed SCA-based algorithm is very close to that of the exhaustive search method based algorithm, which implies that the conservativeness of the proposed approximation method is limited in our simulation.

\begin{figure}[!tp]
  \centering
  \!\!\!\!\!\!\includegraphics[width=0.96\linewidth]{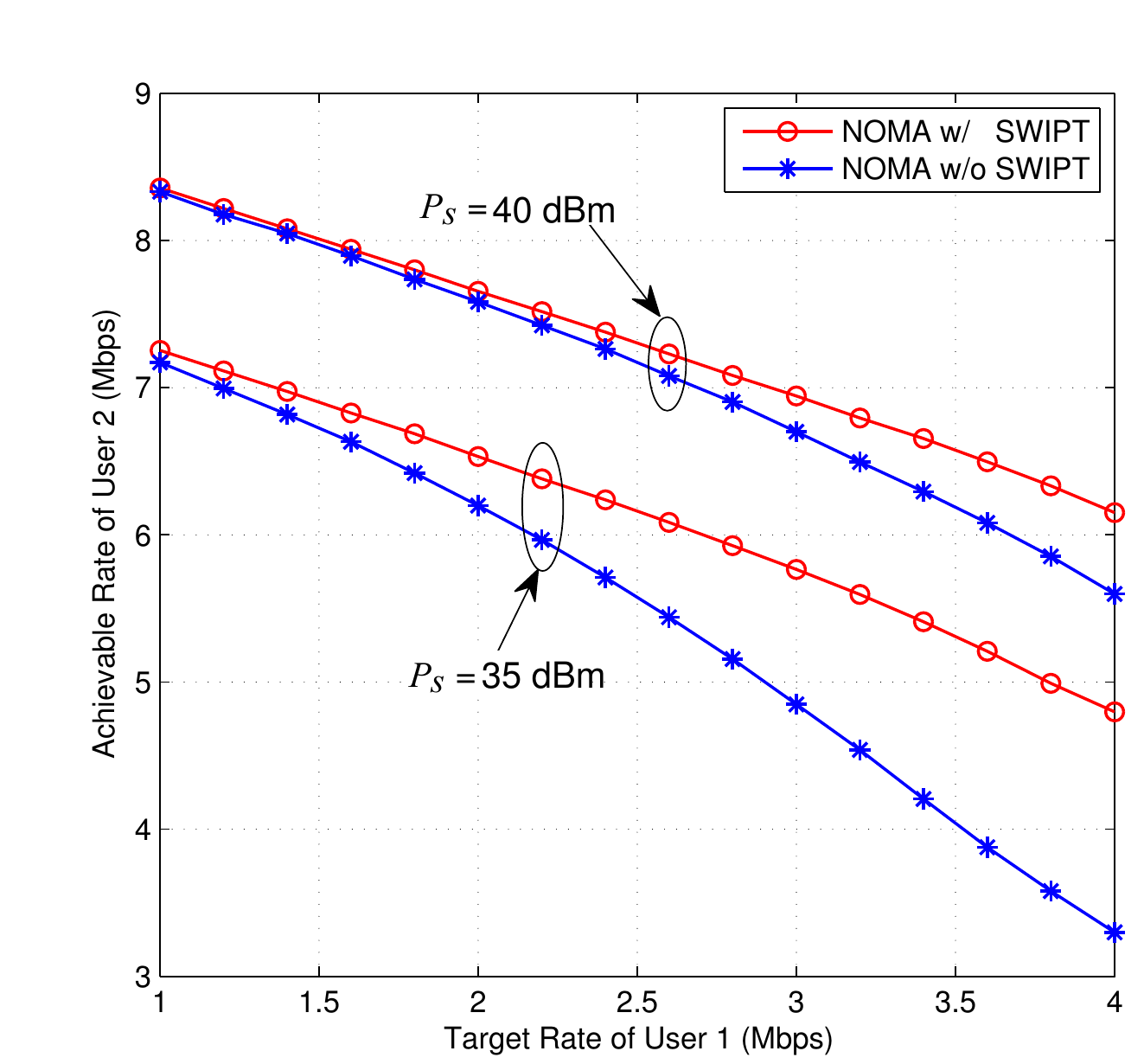}\\
  \caption{Rates tradeoff between user $1$ and user $2$ in MISO cases with $N_t =2$.} \label{tradeoff}
\end{figure}

To investigate the impact of different target rates of user $1$ on the achievable data rate of user $2$, Fig. \ref{tradeoff} shows the rates tradeoff between the two users in two different transmission power cases, i.e., $P_s = 35$ dBm and $40$ dBm, respectively. It is observed that the achievable data rate of user $2$ decreases with the increasing of the target rate of user $1$. The reason is that, with the increasing of the target rate of user $1$, the BS coordinates the beamforming vectors to satisfy the requirement of user $1$, which leads to a reduction of the power allocated to user $2$. Note that the gap between the two transmission schemes increases with the increasing target rate of user $1$ in two transmission power cases (both are high transmission power cases, as depicted in Fig. \ref{comp_miso}), implying that the cooperative SWIPT NOMA works better than the noncooperative NOMA when user $1$ prefers a higher data rate.

\begin{figure}[!tp]
  \centering
  \includegraphics[width=0.96\linewidth]{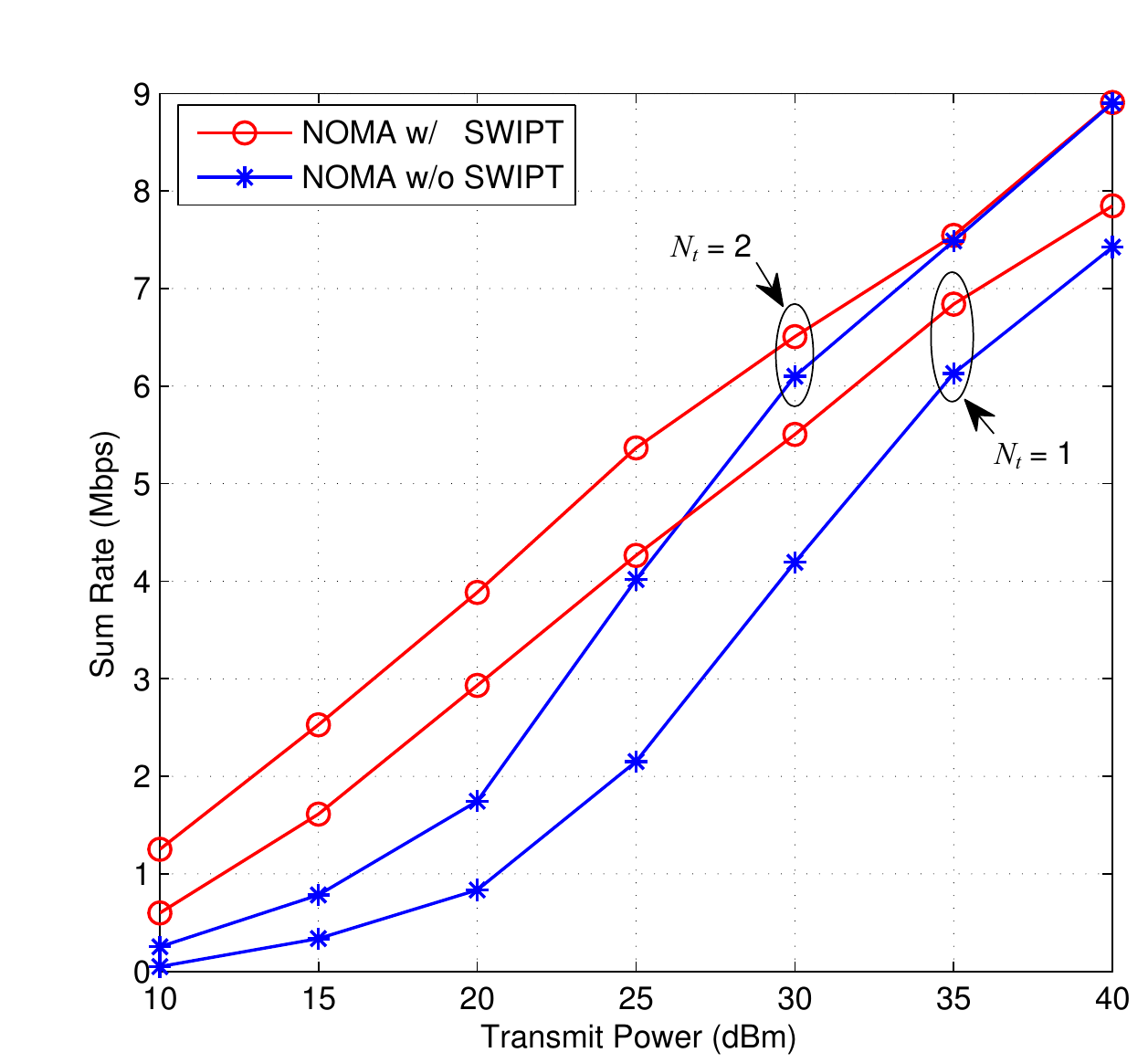}\\
  \!\caption{Performance comparison of different transmission strategies with different antenna number with $\gamma_1 = 1$.}\label{antenna}
\end{figure}
Fig. \ref{antenna} shows the sum rate comparison of different transmission strategies with different numbers of transmission antennas. We observe that the multiple antenna technique can efficiently improve the sum rate of the system in different transmission strategies. Furthermore, we find that the noncooperative NOMA strategy with MISO achieves a better performance than the cooperative SWIPT NOMA with SISO in the high transmission power regime, which also implies the importance of the multi-antenna technique.

\begin{figure}[!h]
  \centering
  \includegraphics[width=0.96\linewidth]{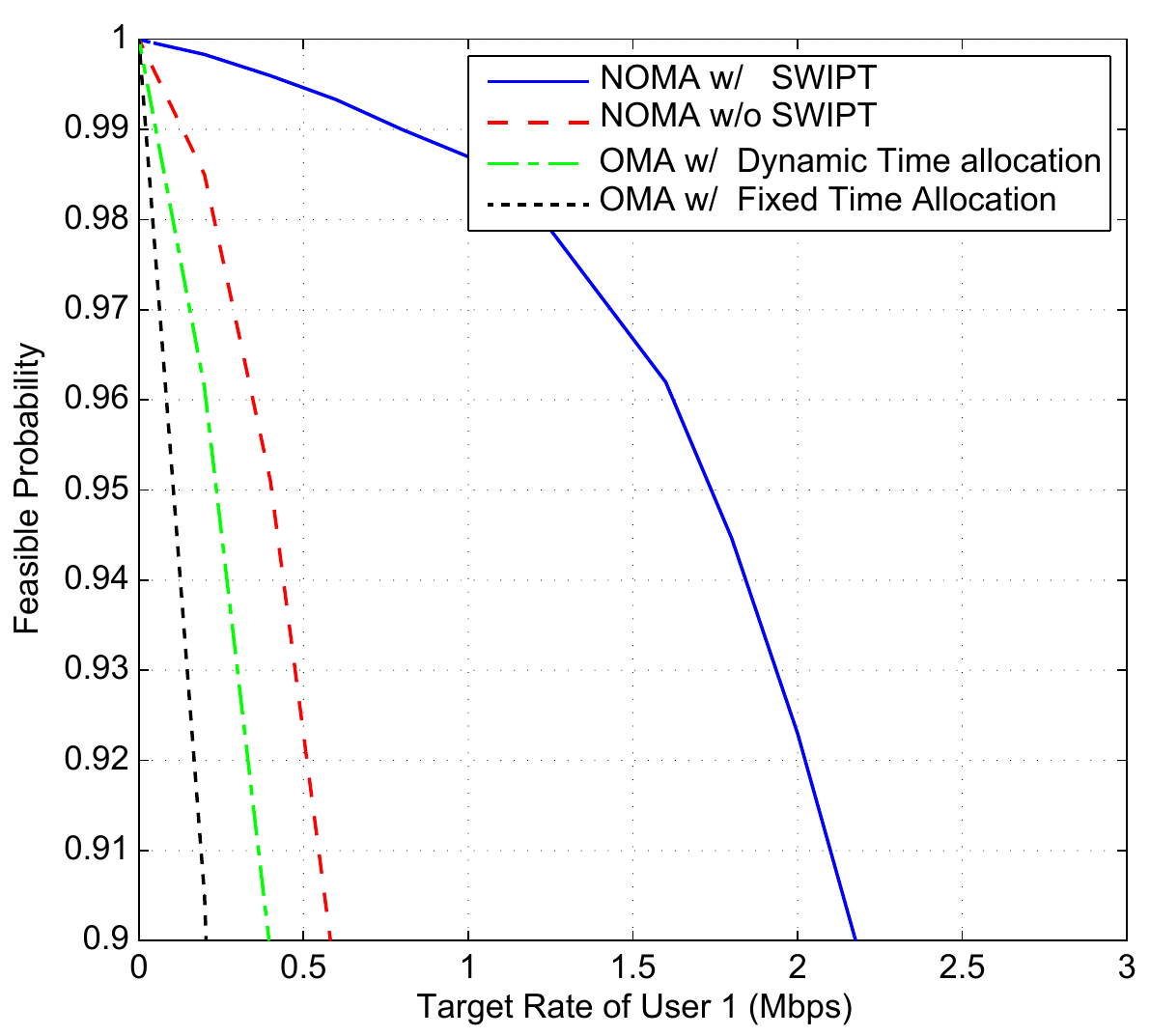}
  \caption{Feasible probabilities of different transmission strategies in MISO cases with $N_t = 2$ and $P_s = 25$ dBm.}
  \label{feasible_gamma}
\end{figure}
Fig. \ref{feasible_gamma} demonstrates the feasible probability versus target rate of user $1$ with different transmission strategies. As can be seen from this figure, with the same feasible probability, the proposed cooperative SWIPT NOMA strategy supports a larger target rate of user $1$ compared with the other considered strategies. Specifically, with a $90\%$ feasible probability, the supported target rate of user $1$ with the cooperative SWIPT NOMA strategy is $2.2$ Mbps, whereas the noncooperative NOMA strategy can only support user $1$ with a target rate of $0.6$ Mbps. From another perspective, with the same target rate of user $1$, the proposed strategy yields the largest feasible probability among all the considered strategies, indicating that the cooperative SWIPT NOMA strategy can significantly improve the communication reliability of user $1$.

\subsection{Optimality of the Obtained Solutions}

\begin{figure}[!tp]
  \centering
  \vspace{3mm}\includegraphics[width=0.96\linewidth]{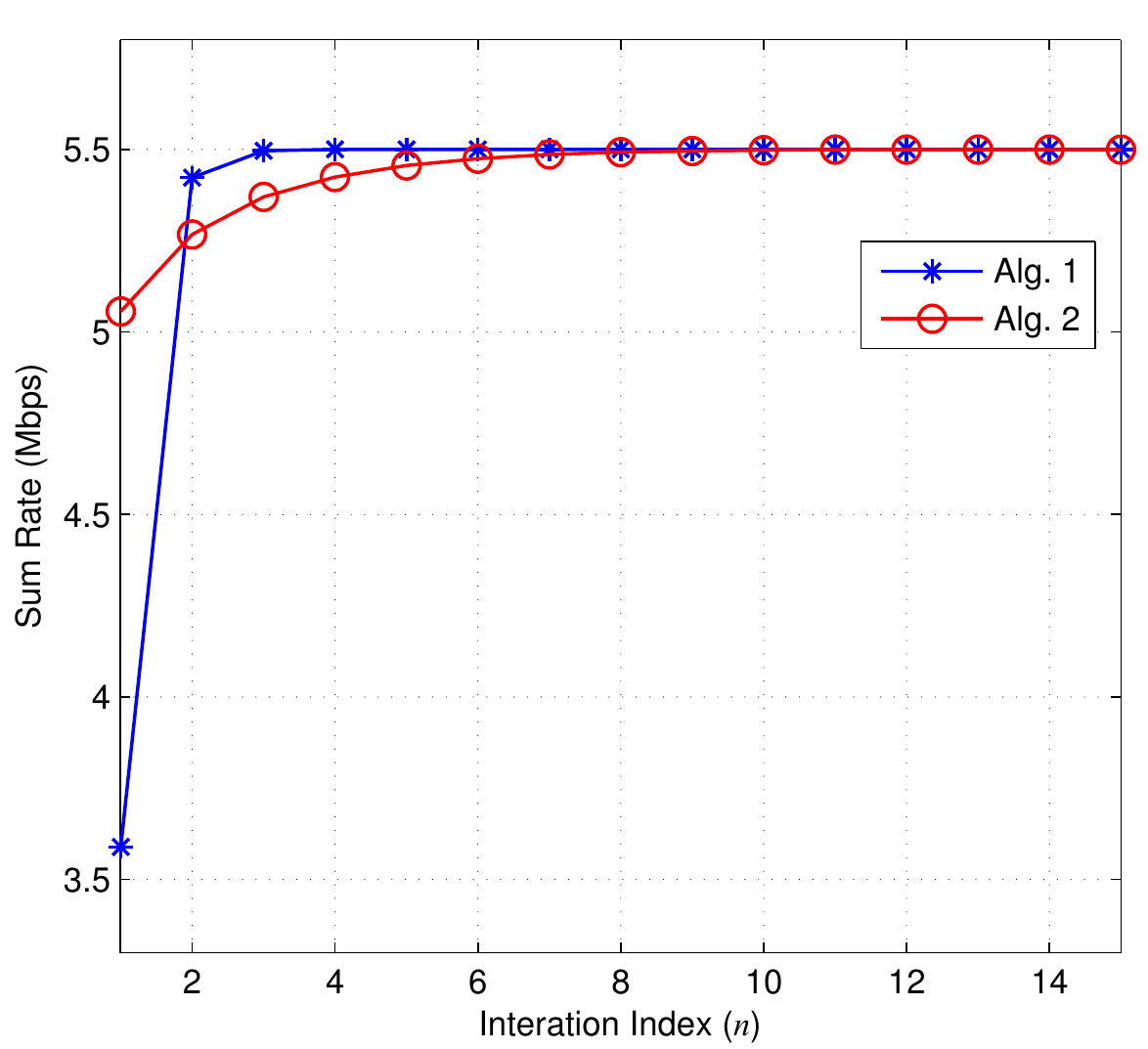}
\caption{Sum rate comparison of SCA and GSS in SISO cases with $P_s = 30$ dBm.}\label{sca_vs_gss}
\end{figure}

Fig. \ref{sca_vs_gss} displays the convergence behavior of the Algorithm \ref{Alg1} and the Algorithm \ref{alg2} in SISO cases. It can be observed that the two algorithms converge to the same sum rate after several iterations, which validates the global optimality of the solution generated by Algorithm \ref{Alg1} in SISO cases.
Thus it motivates us to investigate the solution optimality of Algorithm \ref{Alg1} in MISO cases. Due to the complexity of {\tt P4}, it is extremely challenging for us to prove the obtained beamforming matrices to be rank-one. So, the rank optimality of the obtained solutions is demonstrated with simulation results as presented in Fig. \ref{rank},
where $R_{\lambda}$ is defined as the ratio of the largest eigenvalue and the second largest eigenvalue of beamforming vectors ${\bf W}_1$ and ${\bf W}_2$,
and which depicts the ratio $R_{\lambda}$ versus antenna number for $2,000$ channel realizations. We can see that {\tt P4} can always yield a sufficiently large $R_{\lambda}$ in all simulation cases,  implying that the proposed SCA based algorithm can always attain rank-one solutions when it converges, i.e., the SDR tightness still holds.
\begin{figure}[!tp]
  \centering
  \!\!\!\includegraphics[width=0.96\linewidth]{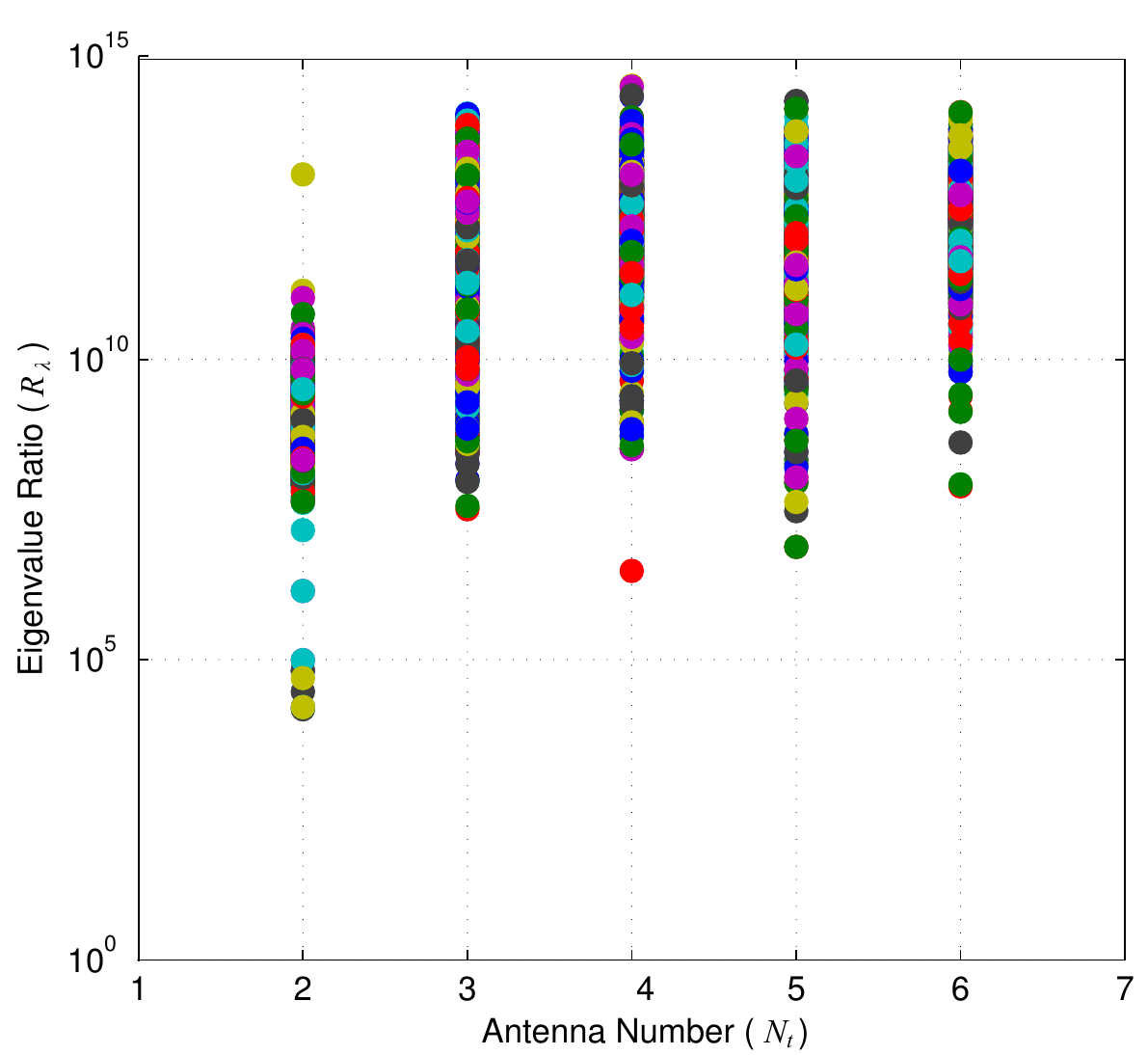}\\
  \caption{Eigenvalue ratio of the obtained solution versus antenna number in MISO cases with $P_s = 30$ dBm.} \label{rank}
\end{figure}


\section{Conclusions}
In this paper, we have proposed a novel cooperative SWIPT-aided NOMA transmission strategy. Two transmission protocols, namely, the cooperative SWIPT NOMA in MISO and SISO cases, have been considered. In MISO cases, the joint design of beamforming and power splitting has been considered. We have equivalently transformed the formulated problem with the SDR technique and proved the rank-one optimality. The reformulated problem can be solved to its global optimal solution by two-dimensional exhaustive search.
However, due to the high complexity of the two-dimensional exhaustive search, an SCA-based algorithm was proposed to efficiently yield at least a stationary point. Motivated by the potential applications, the cooperative SWIPT NOMA protocol design in SISO cases has also been investigated. To obtain the global optimal solution, a GSS-based algorithm was presented by proving the optimal value of the formulated problem is unimodal with respect to the PS ratio. In addition, the optimal solution can be written as a semiclosed-form expression. Moreover, we have found that both the algorithms can converge to the unique global optimal solution in SISO cases. The simulation results have shown that the proposed cooperative SWIPT NOMA strategy outperforms the existing strategies, thus marking it a promising candidate for supporting the functionality of the IoT scenarios.


\begin{appendices}
\section{Proof of Lemma \ref{lem1}}\label{app1}
We prove Lemma \ref{lem1}, i.e., the optimal $\alpha$ makes at least one of the two constraints \eqref{p7.2} and \eqref{p7.3} hold the equality, by contradiction.

Assume $\alpha_0$ and $\beta_0$ are the optimal solution of {\tt P6} and make the constraints \eqref{p7.2} and \eqref{p7.3} hold the inequality.
There must exists a solution $(\alpha^{\prime}<\alpha_0, \beta_0)$ of {\tt P6}, which can guarantee the the constraints satisfied, and meanwhile, decrease the objective function. Thus the solution $(\alpha_0, \beta_0)$ could not be the optimal solution. Therefore, for given $\beta$, the the optimal $\alpha$ will make at least one of the two constraints \eqref{p7.2} and \eqref{p7.3} hold the equality.
This completes the proof. \hfill $\blacksquare$

\section{Proof of Proposition \ref{prop2}}\label{app2}
Here we prove that it is always true that $0 \le B \le 1$ with $\beta \in [\beta_{\min}, \beta_0]$. By Assuming $B \le 1$, we have
\begin{align*}
& \frac{(\gamma_1-\beta h_2 g)(h_1+1)}{(\gamma_1-\beta h_2 g +1)h_1} \le 1\\
\Leftrightarrow  &~ (\gamma_1 - \beta h_2 g)(h_1 + 1) \le (\gamma_1 - \beta h_2 g +1)h_1 \\
\Leftrightarrow  &~ \gamma_1 h_1 +\gamma_1 \!-\! \beta h_1 h_2 g \!-\! \beta h_2 g \le \gamma_1 h_1 \!-\! \beta h_1 h_2 g \!+\! h_1\\
\Leftrightarrow  &~ \beta h_2 g \ge \gamma_1 - h_1
\Leftrightarrow  \beta \ge \frac{\gamma_1 - h_1}{h_2 g}.
\end{align*}
Due to the minimum value of $\beta$ satisfies $\beta_{\min} = \frac{(\gamma_1 - h_1)^+}{h_2 g} \ge \frac{\gamma_1 - h_1}{h_2 g}$. So, for any $\beta \in [\beta_{\min}, \beta_{\max}]$, $B \le 1$ holds true. We observe that $B \ge 0$ always holds true when $\beta \le \frac{\gamma_1}{h_2 g}$. Due to $\beta_0 \le \frac{\gamma_1}{h_2 g}$, we conclude that $0 \le B \le 1$ is always true with $\beta \in [\beta_{\min}, \beta_0]$.
\hfill $\blacksquare$

\begin{figure*}[!tp]
\normalsize
\begin{align}\label{Eq38}
f^{\prime}(\beta) &= \frac{\left(-2h_2g\beta + h_2g+\gamma_1\right)(\gamma_1+1-h_2g\beta) h_1(h_1+1)}{\left(\gamma_1-\beta h_2 g+1\right)^2 h_1^2}
+ \frac{\left(-h_2g\beta^2 + (h_2g+\gamma_1)\beta -\gamma_1\right)(h_1+1)h_1h_2g}{\left(\gamma_1-\beta h_2 g+1\right)^2 h_1^2} - 1 \nonumber\\
&= \frac{\left(2h_2^2g^2\beta^2 -2h_2g(\gamma_1+1)\beta - h_2g(h_2g + \gamma_1)\beta + (h_2g +\gamma_1)(\gamma_1+1)\right) h_1(h_1+1)}{\left(\gamma_1-\beta h_2 g+1\right)^2 h_1^2} \nonumber\\
&\quad + \frac{\left(-h_2^2g^2\beta^2 + h_2g(h_2g+\gamma_1)\beta -h_2g\gamma_1\right)h_1(h_1+1)}{\left(\gamma_1-\beta h_2 g+1\right)^2 h_1^2} - 1 \nonumber\\
&= \frac{\left(h_2^2 g^2 \beta^2 - 2h_2g(\gamma_1+1)\beta +h_2g+\gamma_1^2 + \gamma_1\right) h_1(h_1+1)}{\left(\gamma_1-\beta h_2 g+1\right)^2 h_1^2} - 1. \tag{38}
\end{align}
 \hrulefill
\end{figure*}
\begin{figure*}[!tp]
\begin{align}\label{Eq39}
\Delta(\beta)  &= \left(h_2^2 g^2 \beta^2 - 2h_2g(\gamma_1+1)\beta +h_2g+\gamma_1^2 + \gamma_1\right) h_1(h_1+1) - \left(\gamma_1-\beta h_2 g+1\right)^2 h_1^2 \nonumber\\
&= h_1 h_2^2 g^2 \beta^2 - 2h_1 h_2 g (\gamma_1+1) \beta + (h_2 g + \gamma_1^2 + \gamma_1)h_1(h_1+1) - (\gamma_1+1)^2h_1^2 \nonumber\\
&= h_1(h_2^2 g^2 \beta^2 - 2h_2 g (\gamma_1+1) \beta + (\gamma_1+1)^2 - (\gamma_1+1)^2) + (h_2 g + \gamma_1^2 + \gamma_1)h_1(h_1+1) - (\gamma_1+1)^2h_1^2 \nonumber\\
&= h_1\left(h_2 g \beta - (\gamma_1 + 1)\right)^2 + (h_2 g + \gamma_1^2 + \gamma_1)h_1(h_1+1) - (\gamma_1+1)^2h_1(h_1 + 1) \nonumber \\
&= h_1\left(h_2 g \beta - (\gamma_1 + 1)\right)^2 + h_1\left(h_2g - (\gamma_1+1) \right)(h_1 + 1). \tag{39}
\end{align}
 \hrulefill
\end{figure*}

\section{Proof of Equation \eqref{first-order} }\label{app3}
Here we give the derivation of the first-order derivative of $f(\beta)$, and
$f(\beta)$ is given as below
\begin{align}
f(\beta) &= \frac{(\gamma_1-\beta h_2 g)(h_1+1)(\beta - 1) }{(\gamma_1-\beta h_2 g+1)h_1}- \beta \nonumber\\
         &= \frac{\left(-h_2g\beta^2 + (h_2g+\gamma_1)\beta -\gamma_1\right)(h_1+1) }{(\gamma_1-\beta h_2 g+1)h_1}- \beta. \nonumber
\end{align}
Then, the first-order derivative can be obtained by \eqref{Eq38}.
\hfill  $\blacksquare$

\section{Proof of Equation \eqref{diff}} \label{app4}
It suffices to verify the monotonicity of $f(\beta)$ by the difference between the numerator and denominator of the first term of the right side of \eqref{first-order}. The detailed derivation of the difference is given by \eqref{Eq39}.
\hfill $\blacksquare$
\end{appendices}

\bibliographystyle{IEEEtran}
\bibliography{NOMA_full}

\begin{IEEEbiography}[{\includegraphics[width=1in,height=1.25in,clip,keepaspectratio]{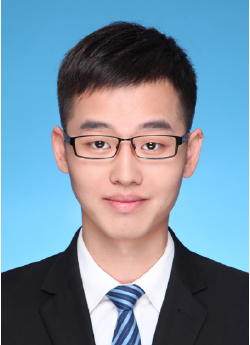}}]{Yanqing Xu}
received the B.S. degree in Electrical Engineering from Hangzhou Dianzi University, Hangzhou, China, in 2014.
He is currently working toward the Ph.D. degree at the State Key Laboratory of
Rail Traffic Control and Safety, Beijing Jiaotong University,  Beijing, China.
From March to September 2017, he was a  Visiting Student at The Chinese University of Hong Kong, Shenzhen, China.
His current research interests include wireless power transfer, non-orthogonal multiple access
and ultra-reliable and low-latency communications.
\end{IEEEbiography}

\begin{IEEEbiography}[{\includegraphics[width=1in,height=1.25in,clip,keepaspectratio]{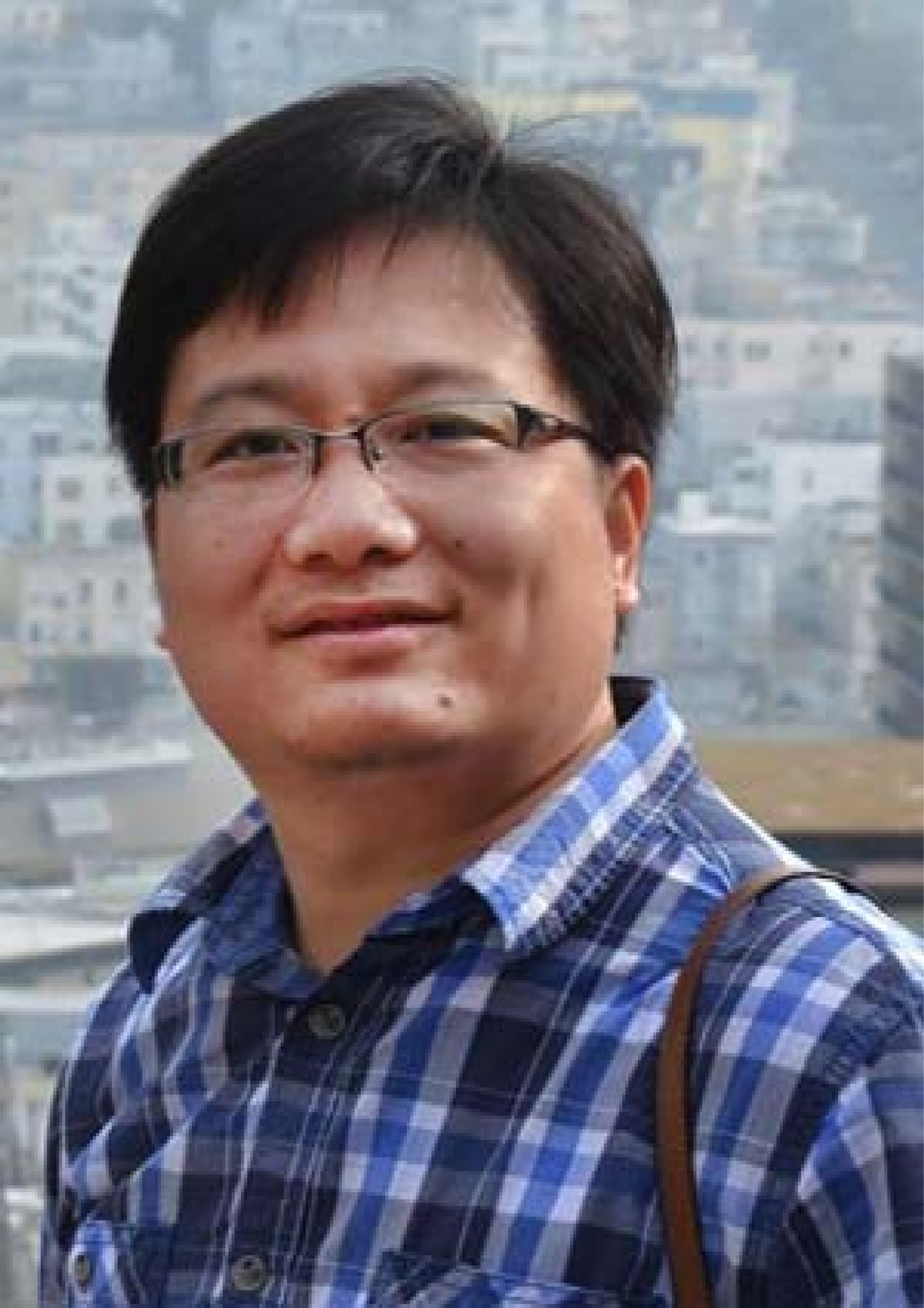}}]{Chao Shen}
(S'12-M'13) received the B.S. degree in communication engineering and the Ph.D. degree in signal and information processing from the Beijing Jiaotong University (BJTU), Beijing, China, in 2003 and 2012, respectively. He was a postdoc at BJTU, and also a Visiting Scholar at the University of Maryland, College Park of Prof. Sennur Ulukus. Since March 2015, he has been with the State Key Laboratory of Rail Traffic Control and Safety, BJTU, Beijing, China, as an Associate Professor.
His current research interests focus on the ultra-reliable and low-latency communications, physical layer design of energy harvesting communications, and energy-efficient wireless communications for high-speed rails.
\end{IEEEbiography}

\begin{IEEEbiography}[{\includegraphics[width=1in,height=1.25in,clip,keepaspectratio]{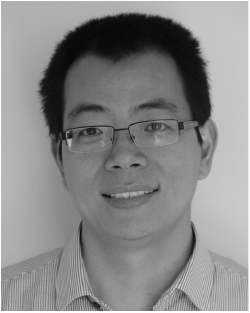}}]{Zhiguo Ding}
(S'03-M'05-SM'15) received his B.Eng in Electrical Engineering from the Beijing University of Posts and Telecommunications in 2000, and the
Ph.D degree in Electrical Engineering from Imperial College London in 2005. From Jul. 2005 to Aug. 2014, he was working in Queen's University Belfast, Imperial College and Newcastle University. Since Sept. 2014, he has been with Lancaster University as a Chair Professor. From Oct. 2012 to Sept. 2016, he has been also with Princeton University as a Visiting Research Collaborator.

Dr Ding's research interests are 5G networks, game theory, cooperative and energy harvesting networks and statistical signal processing.
He is serving as an Editor for \textit{IEEE Transactions on Communications}, \textit{IEEE Transactions on Vehicular Technologies} and \textit{Journal of Wireless Communications and Mobile Computing}. He was an Editor for \textit{IEEE Wireless Communication Letters} and \textit{IEEE Communication Letters} from 2013-2016. He is a Leading Guest Editor for the IEEE JSAC Special Issue on Non-orthogonal Multiple Access for 5G and a Guest Editor for the IEEE Wireless Communication Special Issue on Non-orthogonal Multiple Access. He received the best paper award in IET Comm. Conf. on Wireless, Mobile and Computing, 2009, IEEE Communication Letter Exemplary Reviewer 2012, and the EU Marie Curie Fellowship 2012-2014.
\end{IEEEbiography}

\begin{IEEEbiography}[{\includegraphics[width=1in,height=1.25in,clip,keepaspectratio]{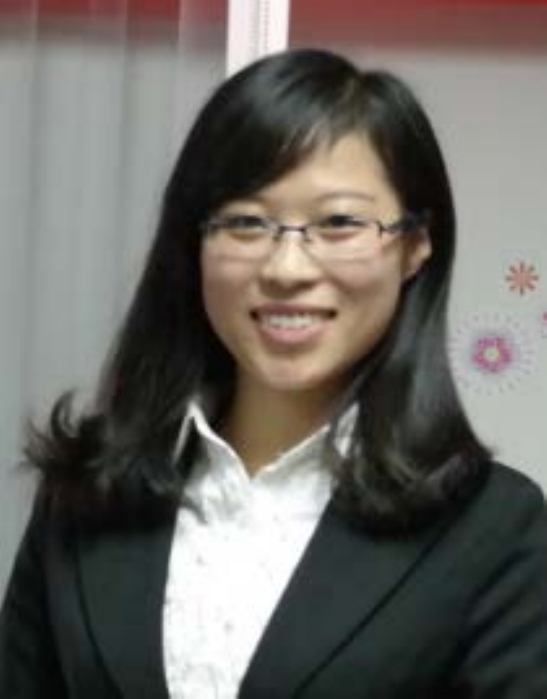}}]{Xiaofang Sun}
(S'14) received the B.S. degree from Beijing Jiaotong University, Beijing, China, in 2013, where she is currently working toward the Ph.D. degree in State Key Laboratory of Rail Traffic Control and Safety.  From September 2014 to July 2015, she was a visiting student at University of Nebraska Lincoln. Since September 2016, she has been a visiting student at the Australian National University, Australia.
Her current research interests include physical-layer optimization for wireless communication networks, and short-packet communications.
\end{IEEEbiography}

\begin{IEEEbiography}[{\includegraphics[width=1in,height=1.25in,clip,keepaspectratio]{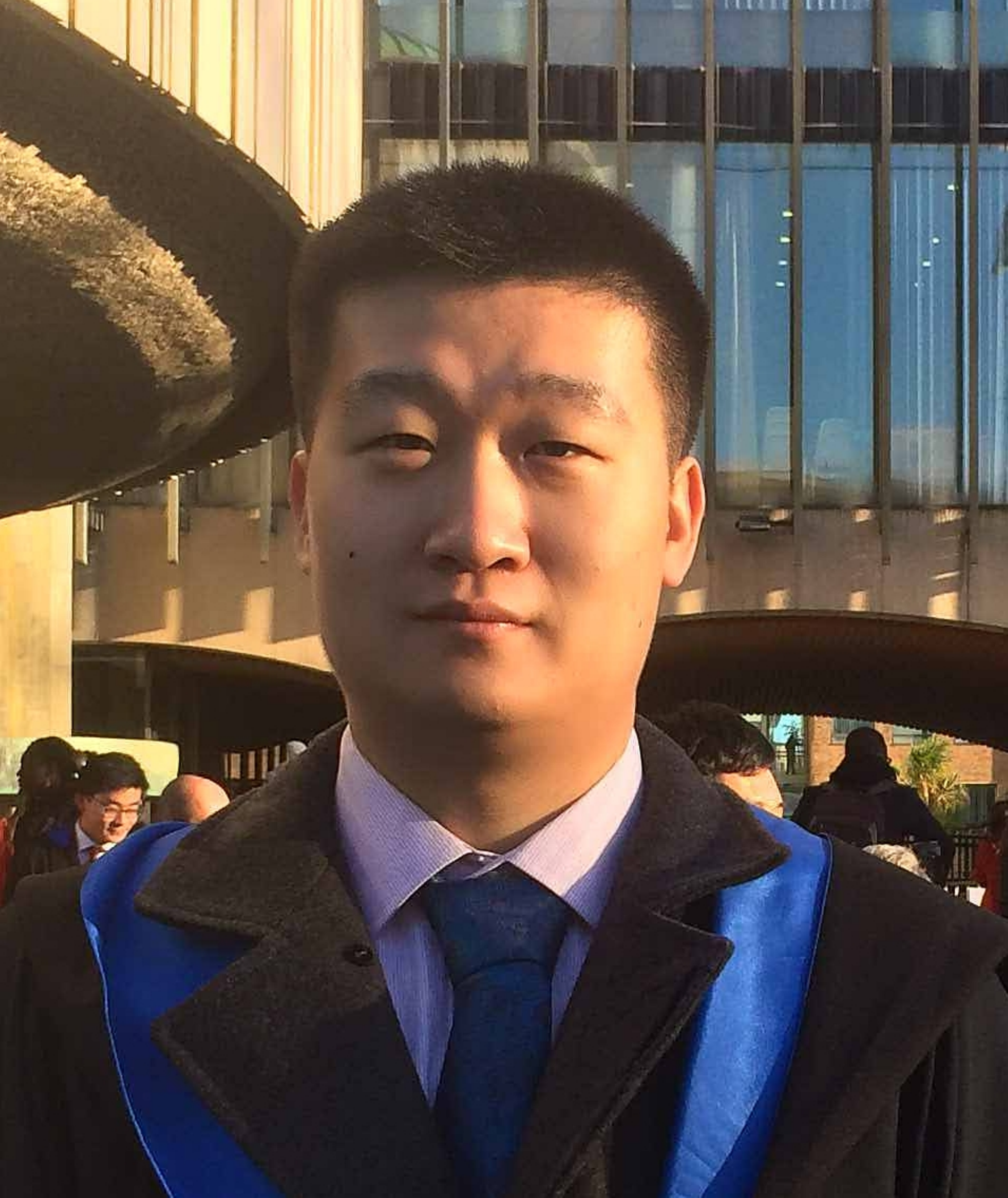}}]{Shi Yan}
received the B.S. degree in Electrical Engineering from Xidian University, Xi'an, China, in 2011 and the M.S. degree in digital signal processing from Newcastle University, Newcastle,U.K, in 2015. He is currently pursuing  the Ph.D. degree at Lancaster University, Lancaster, U.K.  His research interests include non-orthogonal multiple access, wireless power transfer and cooperative energy harvesting networks.
\end{IEEEbiography}

\begin{IEEEbiography}[{\includegraphics[width=1in,height=1.25in,clip,keepaspectratio]{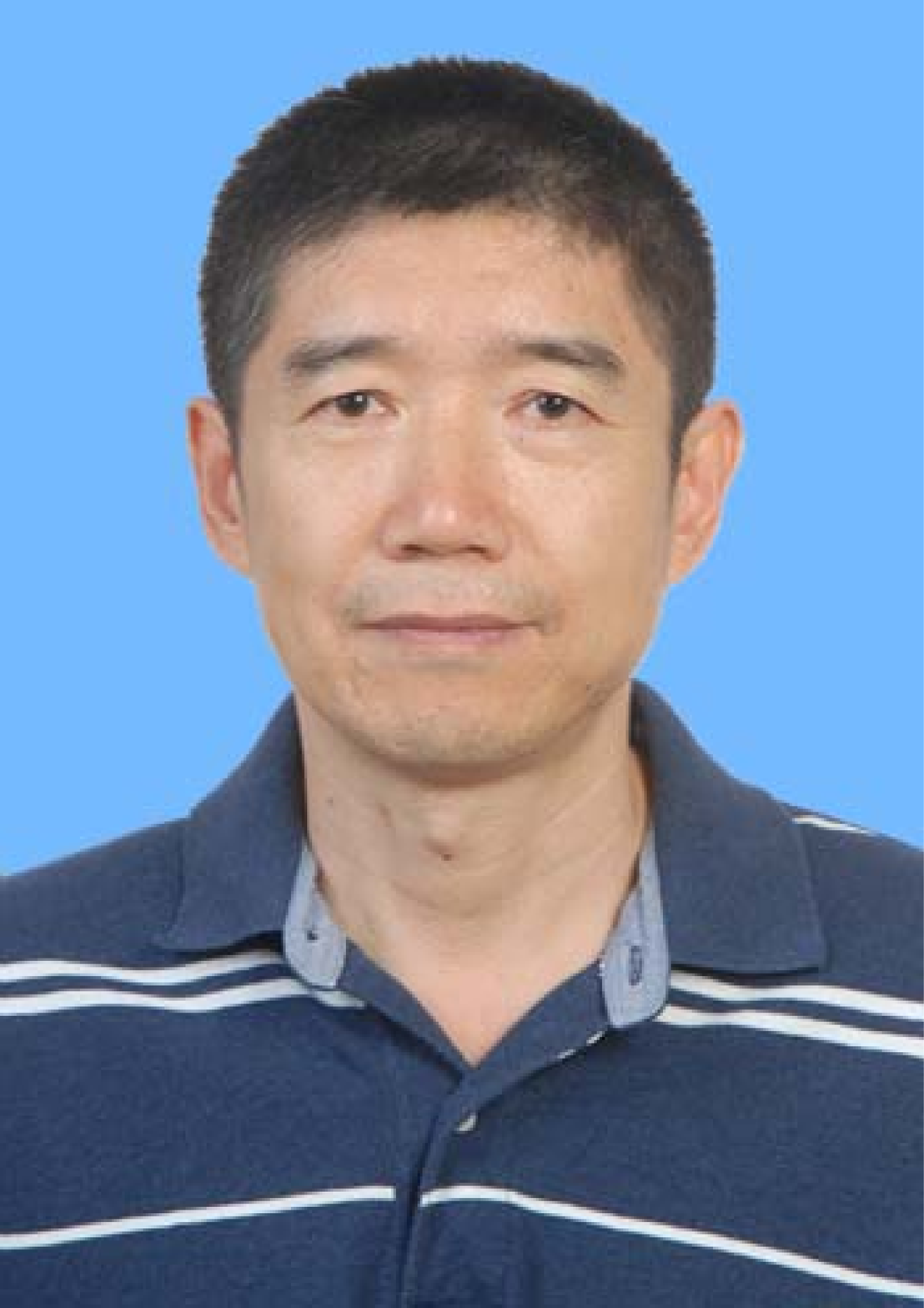}}]{Gang Zhu}
received the M.S. degree in 1993 and PhD degree in 1996 from Xi'an Jiaotong University, Xi'an, China. Since 1996, he has been with Beijing Jiaotong University, Beijing, China, where he is currently a professor. During 2000 to 2001, he was a Visiting Scholar at the
University of Waterloo, Canada.
His current research interests include
resource management in wireless communications, and
GSM for railways (GSM-R). He received Top Ten Sciences and Technology Progress of Universities in China in 2007 and First Class Award of Science and Technology in Railway in 2009.
\end{IEEEbiography}

\begin{IEEEbiography}[{\includegraphics[width=1in,height=1.25in,clip,keepaspectratio]{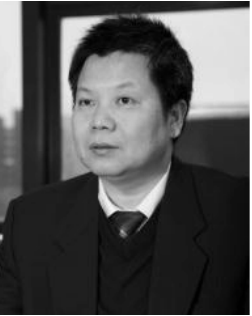}}]{Zhangdui Zhong}
(SM'16) received the B.E. and M.S. degrees from Beijing Jiaotong University, Beijing, China, in 1983 and 1988, respectively.

He is currently a Professor and an Advisor of Ph.D. candidates with Beijing Jiaotong University, where he is also currently a Chief Scientist of State Key Laboratory of Rail Traffic Control and Safety. He is also the Director of the Innovative Research Team of Ministry of Education, Beijing, and a Chief Scientist of Ministry of Railways, Beijing.
He is an Executive Council Member of the Radio Association of China, Beijing, and a Deputy Director of the Radio Association, Beijing. His interests include wireless communications for railways, control theory and techniques for railways, and GSM-R systems. His research has been widely used in railway engineering, such as the Qinghai-Xizang railway, Datong-Qinhuangdao Heavy Haul railway, and many high-speed railway lines in China. He has authored or coauthored seven books, five invention patents, and more than 200 scientific research papers in his research area.

Prof. Zhong received the Mao YiSheng Scientific Award of China, Zhan TianYou Railway Honorary Award of China, and Top 10 Science/Technology Achievements Award of Chinese Universities.
\end{IEEEbiography}

\end{document}